\documentclass[10pt]{article}

\usepackage[utf8]{inputenc}
\usepackage[T1]{fontenc}
\usepackage{amsmath,amssymb,amsthm}
\usepackage{graphicx}
\usepackage{booktabs}
\usepackage[margin=0.85in]{geometry}
\usepackage{natbib}
\usepackage{hyperref}
\hypersetup{hypertexnames=false}
\usepackage{setspace}
\usepackage{float}

\newtheorem{theorem}{Theorem}

\title{Sequential Randomization Tests Using e-values:\\ Applications for trial monitoring}

\author{
  Fernando G. Zampieri$^{1}$
  \\[1em]
  \small $^1$Department of Critical Care Medicine, University of Alberta, Edmonton, Canada \\
}

\date{May 6, 2026}

\begin{document}

\maketitle

% ============================================
% ABSTRACT
% ============================================
\begin{abstract}
Sequential monitoring of randomized trials traditionally relies on parametric assumptions or asymptotic approximations. We discuss a family of nonparametric sequential tests---collectively called e-RT---for binary, event-only, and continuous endpoints. All active variants derive validity from the randomization mechanism. Using a betting framework, each test constructs a test martingale by sequentially wagering on randomized assignments or observed event labels before using the current label in the wealth update. Under the null hypothesis of no treatment effect, the expected wealth cannot grow, guaranteeing anytime-valid Type I error control regardless of stopping rule. The default e-RT posture is effect-size agnostic: monitoring can begin without specifying a hypothesized treatment effect. Alternatively, fixed design-calibrated wagers, including growth-rate-optimal (GROW) wagers, may be used as optional efficiency tools when a clinically meaningful design alternative is credible. We present simulation studies demonstrating calibration and power, and discuss the principled asymmetry in betting strategies across outcome types. These methods provide a conservative, assumption-light complement to model-based sequential analyses.

\end{abstract}

% ============================================
% INTRODUCTION
% ============================================
\section{Introduction}

Sequential monitoring of randomized controlled trials requires methods that control Type I errors regardless of when or why the monitoring stops. Traditional group-sequential designs rely on parametric assumptions and predetermined stopping boundaries. When these assumptions fail, or when trials adapt in ways not fully prespecified, validity guarantees may erode.

In acute-care trials, outcome information can accumulate quickly between scheduled interim analyses. Traditional monitoring often uses $\alpha$-spending functions with prespecified looks; evidence that emerges between those looks may not affect trial decisions until the next analysis.

E-values and e-processes offer an alternative framework \citep{shafer2021, vovk2021, ramdas2025}. An e-value is a measure of evidence against a null hypothesis with a specific property: its expected value under the null is at most 1. This simple constraint yields anytime-valid inference: the Type I error guarantee holds at any stopping time, regardless of the stopping rule.

\citet{pmlr-v177-duan22a} introduced interactive rank testing by betting (i-bet), which tests treatment effects by wagering on treatment assignments given observed outcomes. The intuition is discussed by \cite{ramdas2021game}. Under the null hypothesis, randomization ensures that assignments are independent of outcomes, so no betting strategy can systematically accumulate wealth.

\citet{sokolova2026evalues} recently developed a practitioner-oriented framework for e-value monitoring in adaptive clinical trials, including design-calibrated binary e-processes and an open-source implementation in the \texttt{evalinger} package. Their work emphasizes a design-calibrated view of e-process construction: the betting strategy is chosen to optimize expected evidence growth under a prespecified clinically meaningful alternative. This connects to growth-rate-optimal (GROW) wagering in the broader e-value literature \citep{ramdas2025}, and provides the central external contrast for the present work: e-RT is effect-size agnostic by default, while design-calibrated wagers are treated here as optional efficiency tools rather than as prerequisites for monitoring.

We propose e-RT, an e-value randomization-test framework for prospective monitoring of randomized trials. Like i-bet \citep{pmlr-v177-duan22a}, e-RT uses betting martingales for inference, but monitors as patients enroll, requires no covariates or working models for validity, and can learn default wager policies from accumulating data rather than fixing them by a hypothesized effect size. Design-calibrated wagers are added as optional efficiency modes, not as the defining feature of the method.

We describe three active variants: e-RTb for binary event/no-event outcomes; e-RTe for event-only monitoring (requiring no non-event tracking); and e-RTc for continuous endpoints. All active variants share the same validity proof---the expected wealth multiplier is exactly 1 under the null---but differ in how they translate outcome data into wagers. We also distinguish the randomization-based validity engine from the wager policy. The default adaptive policies preserve effect-size agnosticism. Prespecified design alternatives can be used to choose more aggressive wagers when those assumptions are credible, but this is an efficiency choice rather than a validity requirement. This distinction is especially relevant when comparing sparse event-only updates with dense every-patient updates.

\section{Unified validity argument}

Every active e-RT variant follows the same three-step recipe: first observe only pre-reveal information; then choose a predictable wager; then reveal the randomized label and multiply wealth by a fair-payoff factor. The endpoint-specific sections differ in the signal and wager policy, not in the validity argument.

Let $\mathcal{F}_{k-1}$ denote all information revealed after the first $k-1$ betting updates. Before update $k$, the analyst may observe a signal $S_k$ that does not reveal the random label being bet on. Examples include the current binary outcome before revealing treatment assignment, the fact that a new event has occurred before revealing the event arm, or a continuous outcome before revealing assignment. Let
\begin{equation}
\mathcal{G}_k = \sigma(\mathcal{F}_{k-1}, S_k)
\end{equation}
be the pre-reveal information at update $k$. A wager is valid if it is $\mathcal{G}_k$-measurable and the resulting multiplier is nonnegative. In this manuscript, we call such wagers \emph{predictable}: they may depend on prior data, current pre-reveal signals, randomization probabilities, and prespecified design alternatives, but not on the label that determines whether the wager wins.

This reveal order is a mathematical bookkeeping device. In an open-label trial, treatment assignments may be known operationally before outcomes are observed. The requirement is not that the clinical team be blinded; it is that the algorithm used to choose the current wager must be computable without using the current label or sign being tested. Past assignments may be used by adaptive wager policies after their corresponding wealth updates have already occurred.

In all active variants below, the hidden label is $A_k \in \{0,1\}$ with known null probability $\pi_k = P_0(A_k=1 \mid \mathcal{G}_k)$. A betting fraction $\lambda_k \in [0,1]$ gives multiplier
\begin{equation}
M_k =
\begin{cases}
\lambda_k / \pi_k, & A_k=1,\\
(1-\lambda_k)/(1-\pi_k), & A_k=0.
\end{cases}
\label{eq:unified_assignment_multiplier}
\end{equation}

\begin{theorem}[Unified e-RT validity]
If each wager is predictable in the sense above and each multiplier is nonnegative, then the wealth process
\begin{equation}
W_k = \prod_{\ell=1}^{k} M_\ell, \qquad W_0=1,
\end{equation}
is a nonnegative test martingale under the null hypothesis. Consequently, for any stopping time $\tau$ and any $\alpha \in (0,1)$,
\begin{equation}
P_0\left(\sup_{k \geq 1} W_k \geq \frac{1}{\alpha}\right) \leq \alpha.
\end{equation}
\end{theorem}

\begin{proof}
For the assignment-prediction multiplier in Equation~\ref{eq:unified_assignment_multiplier},
\begin{align}
\mathbb{E}_0(M_k \mid \mathcal{G}_k)
&= \pi_k \frac{\lambda_k}{\pi_k}
 + (1-\pi_k)\frac{1-\lambda_k}{1-\pi_k} \\
&= \lambda_k + (1-\lambda_k) = 1.
\end{align}
Therefore,
\begin{equation}
\mathbb{E}_0(W_k \mid \mathcal{G}_k)
= W_{k-1}\mathbb{E}_0(M_k \mid \mathcal{G}_k)
= W_{k-1}.
\end{equation}
Taking conditional expectation again with respect to $\mathcal{F}_{k-1}$ gives $\mathbb{E}_0(W_k \mid \mathcal{F}_{k-1})=W_{k-1}$, so $(W_k)$ is a martingale. Nonnegativity follows from the multiplier constraint. Ville's inequality for nonnegative martingales \citep{ville1939} then gives the anytime-valid Type I error bound.
\end{proof}

This theorem is deliberately agnostic about how the wager is selected. An adaptive wager is valid if it is computed from past data and current pre-reveal information. A fixed or design-calibrated wager is valid if it is prespecified from protocol quantities or a design alternative before the trial begins, and then applied without looking at the current hidden label. Misspecifying the design alternative can reduce power, delay crossing, or inflate the apparent effect among crossing trials, but it does not change the conditional expectation calculation above. Validity is a property of randomization, exchangeability, predictability, and nonnegative fair multipliers; power is a property of how well the wager matches the alternative. The stopping-time guarantee covers stopping rules based on the revealed monitoring history and other predictable trial information. Response-adaptive randomization can also be accommodated if the conditional randomization probabilities used in the multiplier are known at each update.

Table~\ref{tab:variant_selection} summarizes the practical endpoint choice before the variants are developed in detail.

\begin{table}[htbp]
\centering
\small
\caption{Endpoint-oriented choice among active e-RT variants.}
\label{tab:variant_selection}
\begin{tabular}{p{0.27\textwidth}p{0.17\textwidth}p{0.46\textwidth}}
\toprule
Trial setting & Suggested variant & Main planning notes \\
\midrule
Binary endpoint with follow-up for events and non-events & e-RTb & Default binary monitor. Adaptive half-Kelly is effect-size agnostic; fixed or design-calibrated wagers are optional when a credible risk difference is available. \\
Reliable event stream but incomplete non-event ascertainment & e-RTe & Useful when events are captured well and non-events are costly to verify. Requires enough expected events for burn-in and ramp; most attractive at low-to-moderate baseline risk with modest ARR. \\
Continuous endpoint & e-RTc & Adaptive sign-direction wagering is assumption-light but conservative for small effects. A normal-shift design wager can improve power when the design effect and scale are credible. \\
\bottomrule
\end{tabular}
\end{table}

\section{e-RT binary (e-RTb)}

Binary endpoints are common in clinical trials, with mortality, clinical deterioration, infection, treatment failure, or response status often recorded as event/no-event outcomes. The binary e-RT, abbreviated \textbf{e-RTb}, is the simplest member of the family and serves as the prototype for the later variants.

\subsection{Setup}

Consider a sequential randomized trial with 1:1 allocation. At each enrollment $i = 1, 2, \ldots$, we observe:
\begin{itemize}
    \item $T_i \in \{0, 1\}$: treatment assignment (0 = control, 1 = intervention)
    \item $Y_i \in \{0, 1\}$: binary outcome (0 = no event, 1 = event)
\end{itemize}
Treatment is assigned with known probability $p = P(T_i = 1)$, typically $p = 0.5$.

The null hypothesis is that treatment assignment has no effect on outcome:
\begin{equation}
H_0: Y_i \perp T_i \text{ for all } i
\end{equation}
Under this hypothesis, observing the outcome provides no information about which arm the patient was assigned to.

\subsection{Wealth Process}

Following \citet{pmlr-v177-duan22a}, we construct a wealth process by wagering on treatment assignments. After observing outcome $Y_i$ but \emph{before} learning treatment assignment $T_i$, we choose $\lambda_i \in [0, 1]$: the fraction wagered on intervention.

The wealth updates as:
\begin{equation}
W_i = W_{i-1} \times
\begin{cases}
\lambda_i / p & \text{if } T_i = 1 \\
(1 - \lambda_i) / (1-p) & \text{if } T_i = 0
\end{cases}
\end{equation}
starting from $W_0 = 1$. When we bet toward the correct arm, wealth grows; when wrong, it shrinks.

\subsection{Betting Strategy}

The wager is chosen before the treatment assignment for patient $i$ is used in the wealth update. It may depend on the current outcome $Y_i$ and on prior patients' outcomes and assignments, but not on $T_i$ itself. The default strategy learns the treatment effect from accumulating data.

Let:
\begin{equation}
\hat{\delta}_{i-1} = \text{(event rate in intervention)} - \text{(event rate in control)}
\end{equation}
estimated from patients $1, \ldots, i-1$. The betting fraction is:
\begin{equation}
\lambda_i =
\begin{cases}
0.5 + 0.5 \cdot c_i \cdot \hat{\delta}_{i-1} & \text{if } Y_i = 1 \\
0.5 - 0.5 \cdot c_i \cdot \hat{\delta}_{i-1} & \text{if } Y_i = 0
\end{cases}
\end{equation}
where $c_i \in [0, 1]$ ramps from 0 to 1 over a burn-in period:
\begin{equation}
c_i = \min\left(1, \max\left(0, \frac{i - n_0}{n_r}\right)\right)
\end{equation}
with $n_0$ the burn-in period and $n_r$ the ramp period. This prevents large bets when $\hat{\delta}_{i-1}$ is unstable due to small samples.

The logic: if $\hat{\delta} > 0$ (more events in intervention), then events suggest intervention and non-events suggest control. If $\hat{\delta} < 0$ (fewer events in intervention), then events suggest control and non-events suggest intervention. The factor of $0.5$ before $c_i \cdot \hat{\delta}_{i-1}$ ensures $\lambda_i \in [0, 1]$.

\subsection{Worked Example}

Consider a trial comparing intervention versus control, with a binary outcome (event or no event). Event is mortality, which is expected to be lower with intervention. Allocation is 1:1 ($p = 0.5$). Assume burn-in is complete ($c_i = 1$).

We look back at patients 1--199. Intervention arm has 100 patients, 35 events, so rate = 35.0\%. Control arm has 99 patients, 40 events, so rate = 40.4\%. $\hat{\delta}_{199} = 0.350 - 0.404 = -0.054$ (intervention looks protective).

Patient 200 has an event (dies). Where is this patient likely from? Events are more common in control (40.4\% vs 35.0\%), so probably control. We bet $\lambda = 0.5 + 0.5 \times (-0.054) = 0.473$ on intervention and $1 - \lambda = 0.527$ on control. Assignment revealed: control. We guessed right. Multiplier: $0.527 / 0.5 = 1.054$. Wealth grows 5.4\%.

Patient 201: Update counts---intervention has 100 patients, 35 events (35.0\%); control has 100 patients, 41 events (41.0\%). $\hat{\delta}_{200} = -0.060$. Patient 201 has no event. Non-events are more common in intervention (65.0\% vs 59.0\%), so probably intervention. Bet: $\lambda = 0.5 - 0.5 \times (-0.060) = 0.530$. Assignment revealed: intervention. Multiplier: $0.530 / 0.5 = 1.060$. Wealth grows 6.0\%.

Patient 202: Update counts---intervention has 101 patients, 35 events (34.7\%); control has 100 patients, 41 events (41.0\%). $\hat{\delta}_{201} = -0.063$. Patient 202 has an event. Bet toward control: $\lambda = 0.5 + 0.5 \times (-0.063) = 0.469$. Assignment revealed: intervention. Wrong guess. Multiplier: $0.469 / 0.5 = 0.938$. Wealth \textbf{shrinks} 6.2\%.

Cumulative wealth:
\begin{equation}
W_{202} = W_{199} \times 1.054 \times 1.060 \times 0.938 = W_{199} \times 1.048
\end{equation}

Despite one wrong guess, wealth grew 4.8\% over these three patients. Under the alternative, correct guesses outnumber incorrect ones on average and wealth grows. Under the null, right and wrong guesses balance out and wealth fluctuates around 1.

\subsection{Validity}

The binary case is the simplest concrete instance of the unified validity argument. Here the hidden label is the current treatment assignment, $A_i=T_i$, and the null probability is $\pi_i=p$. Under the null, outcome and treatment are independent, so after observing $Y_i$ and all prior data, the current treatment assignment still has probability $P(T_i=1)=p$. For any predictable wager $\lambda_i$,
\begin{align}
\mathbb{E}[\text{multiplier} \mid \mathcal{G}_i] &= p \times \frac{\lambda_i}{p} + (1-p) \times \frac{1 - \lambda_i}{1-p} \\
&= \lambda_i + (1 - \lambda_i) = 1
\end{align}
Thus $(W_i)$ is a nonnegative martingale under the null, and Ville's inequality gives anytime-valid Type I error control when rejecting at threshold $1/\alpha$.

\subsection{Simulation studies}

We evaluated operating characteristics of e-RTb by simulation. For each scenario, we calculated the sample size required for a chi-square test to achieve the target power at $\alpha = 0.05$, then ran 5{,}000 simulated trials at that sample size. We used burn-in = 50 patients and ramp = 100 patients. Control arm event rate was 40\% in all scenarios.

\subsubsection{Results}

Table~\ref{tab:simulations} presents Type I error and power for trials designed to detect 5\% or 10\% absolute risk reductions (ARR) with 80\% or 90\% power.

\begin{table}[htbp]
\centering
\caption{Operating characteristics for adaptive e-RTb. Each row summarizes 5{,}000 fixed-seed simulated trials under the null and 5{,}000 under the matched alternative.}
\begin{tabular}{@{}cccccc@{}}
\toprule
ARR & Target Power & $N$ & Type I Error & e-RTb Power & Median Crossing \\
\midrule
5\% & 80\% & 2{,}942 & 0.031 & 47.5\% & 1{,}450 (49\%) \\
10\% & 80\% & 712 & 0.021 & 49.5\% & 401 (56\%) \\
5\% & 90\% & 3{,}938 & 0.035 & 63.6\% & 1{,}837 (47\%) \\
10\% & 90\% & 954 & 0.025 & 64.9\% & 479 (50\%) \\
\bottomrule
\end{tabular}
\label{tab:simulations}
\end{table}

Type I error was controlled below the nominal $\alpha = 0.05$ level across the tested scenarios, consistent with the martingale guarantee. Power was lower than the corresponding fixed-sample design power because the e-process must cross an anytime-valid threshold. When the process rejected the null, it generally did so around the middle of planned enrollment.

\subsubsection{Trajectory examples}

Figure~\ref{fig:null} shows representative null trajectories. Wealth fluctuates near 1 and no path crosses the threshold in these panels.

\begin{figure}[htbp]
\centering
\includegraphics[width=0.48\textwidth]{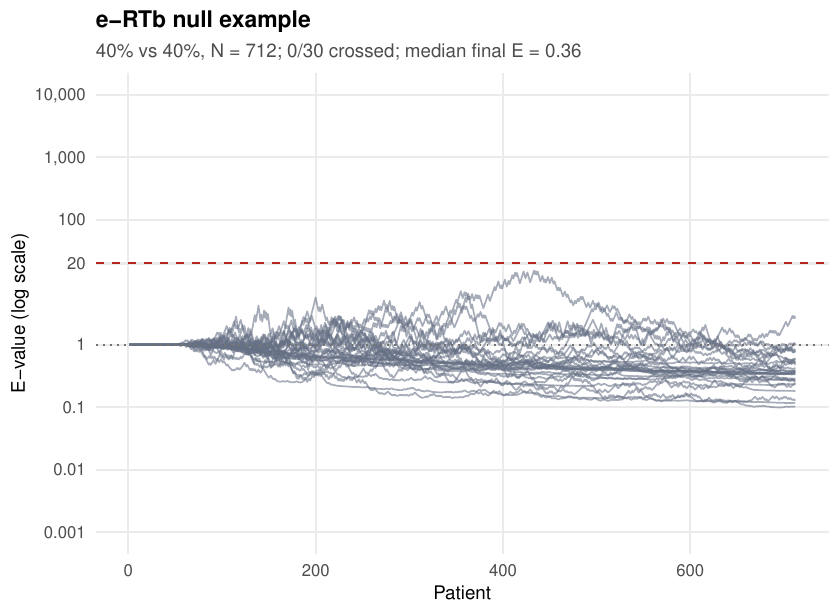}
\includegraphics[width=0.48\textwidth]{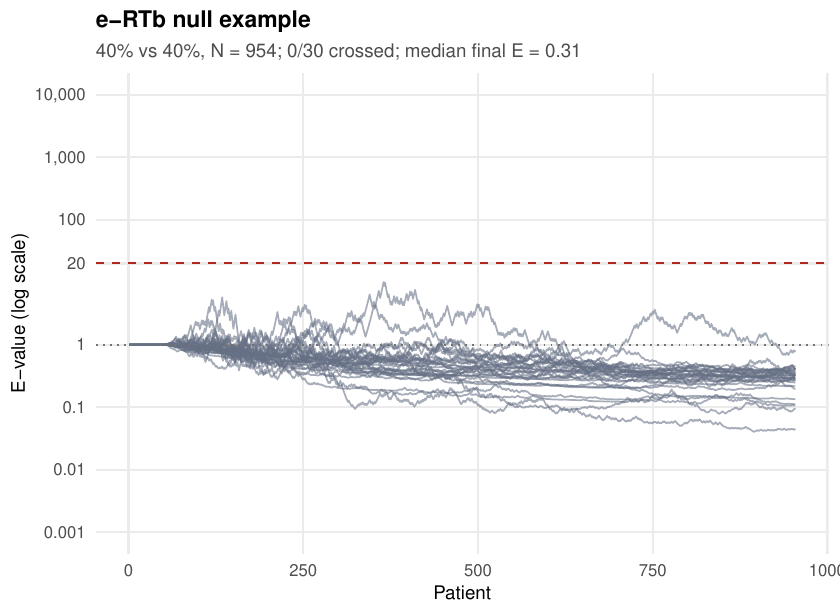}
\caption{Wealth trajectories under the null hypothesis. Left: n = 712 (80\% power design). Right: n = 954 (90\% power design). Dashed red line: rejection threshold ($1/\alpha = 20$). Dotted gray line: neutral (wealth = 1). Under the null, no trajectory crosses the threshold.}
\label{fig:null}
\end{figure}

Figure~\ref{fig:alt} shows representative trajectories under a 10pp ARR. In these panels, 15 of 30 paths cross in the 80\% power design and 20 of 30 cross in the 90\% power design, usually before enrollment completes.

\begin{figure}[htbp]
\centering
\includegraphics[width=0.48\textwidth]{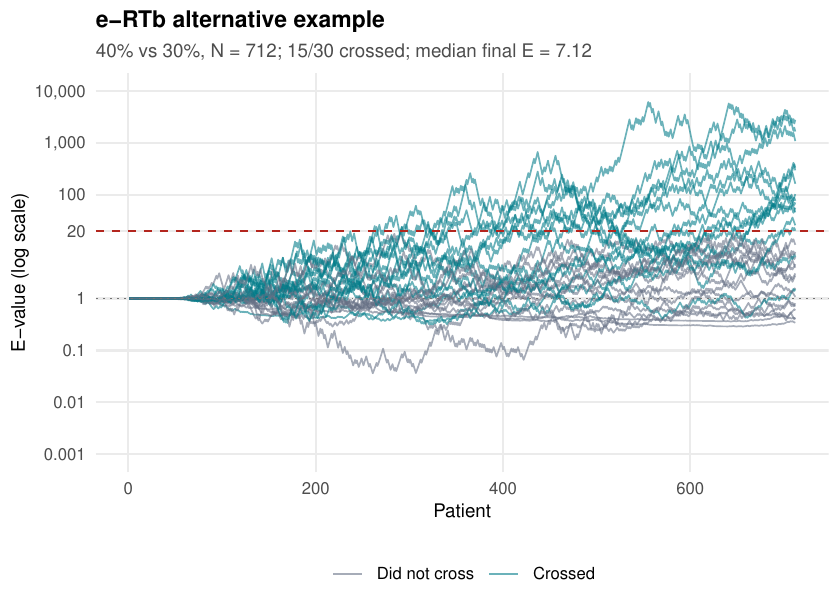}
\includegraphics[width=0.48\textwidth]{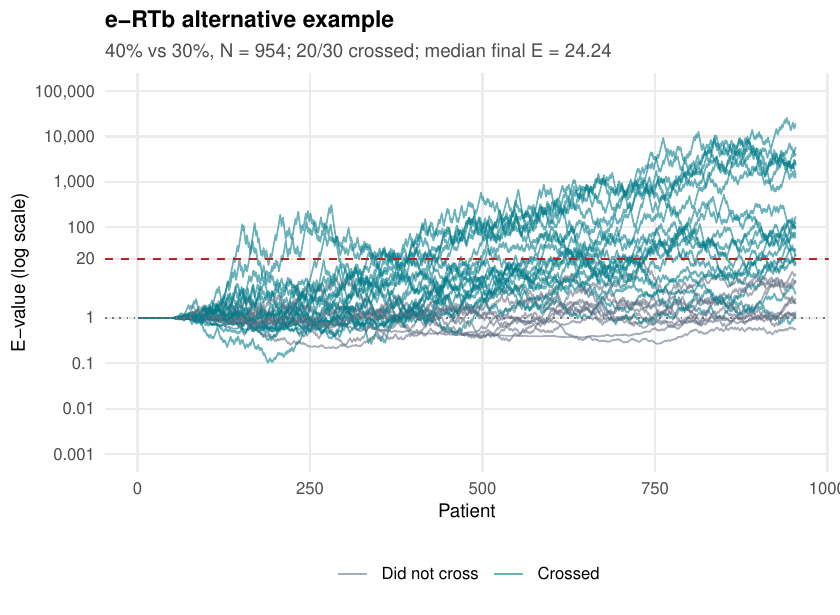}
\caption{E-process trajectories under the alternative hypothesis (true ARR = 10\%). Left: n = 712 (80\% power design). Right: n = 954 (90\% power design). Dashed red line: rejection threshold ($1/\alpha = 20$). Under the alternative, approximately half to two thirds of the representative trajectories cross the threshold, typically around the midpoint of enrollment.}
\label{fig:alt}
\end{figure}

\section{Event-Only Monitoring (e-RTe)}

\subsection{Motivation}

The e-RTb requires knowing both the treatment arm and the outcome for every enrolled patient. In practice, this means a coordinator must ascertain whether each patient experienced the event or not---requiring follow-up, data entry, and outcome adjudication for all patients, including the majority who do not experience the event.

In some settings, events of interest are reliably captured but non-events require active follow-up. Examples include mortality in electronic medical records, oncology progression, or cardiovascular events. This motivates a simpler variant: \textbf{e-RTe monitors only the stream of events, ignoring non-events entirely.}

\subsection{Null Hypothesis}

Let $A_i$ denote the arm label of the $i$th observed event, with $A_i=1$ for treatment and $A_i=0$ for control. Let $\mathcal{H}^e_i$ be the event-monitoring history just before the arm of event $i$ is revealed, including the fact that event $i$ has been observed but not its arm. The simple e-RTe version used here assumes that, under the null, the observed event labels are conditionally exchangeable between arms. In a 1:1 randomized trial this means
\begin{equation}
H_0: P(A_i=1 \mid \mathcal{H}^e_i) = 0.5 .
\end{equation}
This condition requires the observed event stream itself to be arm-exchangeable under the null: for example, one relevant event per patient, arm-exchangeable event opportunity, and arm-exchangeable event ascertainment. Equal enrollment alone is not sufficient if follow-up, detection, censoring, or stratum-specific allocation makes the next observed event more likely to come from one arm under the null.

Thus e-RTe reduces monitoring to a sequential Bernoulli test of the ``event coin'': among observed events, is the treatment-arm fraction equal to $0.5$? More generally, if the null event-label probability is a known predictable value $\pi_i$, the same construction remains valid after replacing $0.5$ by $\pi_i$ in the wealth multiplier.

\subsection{Algorithm}

The algorithm maintains two counters: $d_{\text{trt}}$ and $d_{\text{ctrl}}$, both initialized to zero. For each event $i = 1, 2, \ldots$:

\textbf{Step 1: Estimate the coin bias from past data.} Using all events before the current one, compute the plug-in estimator:
\begin{equation}
\hat{p}_{i-1} = \frac{d_{\text{trt}}}{d_{\text{trt}} + d_{\text{ctrl}}}
\end{equation}
If no events have been observed yet, set $\hat{p} = 0.5$. Under the null, this running treatment-event fraction should hover around $0.5$; if treatment is protective, it should drift below $0.5$.

\textbf{Step 2: Compute the wager.} The betting fraction is:
\begin{equation}
\lambda_i = 0.5 + c_i \cdot (\hat{p}_{i-1} - 0.5)
\end{equation}
clamped to $[0.001, 0.999]$, where $c_i$ is the same ramp function as before:
\begin{equation}
c_i = \min\left(1, \max\left(0, \frac{i - n_0}{n_r}\right)\right)
\end{equation}
with burn-in $n_0 = 30$ events and ramp $n_r = 50$ events. During burn-in the wager is neutral; after the ramp completes, $\lambda_i = \hat{p}_{i-1}$. For example, if 40\% of past events came from treatment, the next wager places 40\% on ``treatment event'' and 60\% on ``control event.'' At full ramp this is the full Kelly event-coin wager. It is more aggressive than the half-Kelly default in e-RTb because e-RTe compounds only at events rather than at every randomized patient.

\textbf{Step 3: Update wealth.}
\begin{equation}
W_i = W_{i-1} \times
\begin{cases}
\lambda_i / 0.5 & \text{if the event is from the treatment arm} \\
(1 - \lambda_i) / 0.5 & \text{if the event is from the control arm}
\end{cases}
\end{equation}
starting from $W_0 = 1$.

This is the fair $2:1$ payout for a $0.5$-probability event. Wealth grows when the wager leans toward the observed arm and shrinks otherwise.
With a known predictable null event-label probability $\pi_i$, the fair multipliers are $\lambda_i/\pi_i$ and $(1-\lambda_i)/(1-\pi_i)$ instead.

\textbf{Step 4: Update counters after betting.} Increment $d_{\text{trt}}$ or $d_{\text{ctrl}}$ depending on which arm the event came from. This ordering---bet first, then update---ensures that the estimate $\hat{p}_{i-1}$ uses only past information, maintaining the martingale property.

\subsection{Worked example}

Consider a trial comparing a new treatment against standard of care in the ICU, where the event of interest is mortality. Suppose 80 deaths have been observed so far. Burn-in (30) and ramp (50) are complete, so $c_i = 1$. Current counts: $d_{\text{trt}} = 33$, $d_{\text{ctrl}} = 47$.

Event 81 arrives. Let us walk through the update.

\textbf{Step 1:} $\hat{p}_{80} = 33 / (33 + 47) = 33/80 = 0.4125$. Interpretation: so far, 41.25\% of events came from the treatment arm. This is below 50\%, suggesting treatment may be protective.

\textbf{Step 2:} $\lambda_{81} = 0.5 + 1.0 \times (0.4125 - 0.5) = 0.4125$. We place 41.25\% of our bet on ``treatment event'' and 58.75\% on ``control event.'' We are leaning toward control because past evidence suggests treatment events are less frequent.

\textbf{Step 3:} The arm is revealed: it is a \textbf{control} event.
\begin{equation}
\text{Multiplier} = \frac{1 - 0.4125}{0.5} = \frac{0.5875}{0.5} = 1.175
\end{equation}
Wealth grows by 17.5\%. Our bet was correct---we leaned toward control, and indeed it was a control event.

\textbf{Step 4:} Update $d_{\text{ctrl}} = 48$. Now $\hat{p}_{81} = 33/81 = 0.407$.

Had the event been from the treatment arm instead:
\begin{equation}
\text{Multiplier} = \frac{0.4125}{0.5} = 0.825
\end{equation}
Wealth would have shrunk by 17.5\%. Our lean toward control would have been wrong.

Under the null, treatment and control events arrive with equal probability, so wins and losses balance on average. Under the alternative ($\hat{p} < 0.5$), control events are genuinely more frequent, and wealth grows systematically.

\subsection{Validity}

\begin{theorem}
If the observed event labels are conditionally exchangeable under the null and each wager is chosen before revealing the current event arm, then the e-RTe wealth process $(W_i)$ is a nonnegative martingale.
\end{theorem}

\begin{proof}
Condition on $\mathcal{H}^e_i$ and the current wager $\lambda_i$. Under the null used here, the probability that event $i$ is from the treatment arm is $0.5$. The expected wealth multiplier is:
\begin{align}
\mathbb{E}\left[\frac{W_i}{W_{i-1}} \,\middle|\, \lambda_i\right] &= 0.5 \times \frac{\lambda_i}{0.5} + 0.5 \times \frac{1 - \lambda_i}{0.5} \\
&= \lambda_i + (1 - \lambda_i) = 1
\end{align}
This is the same identity as in the binary case. Since the expected multiplier is exactly 1, wealth cannot systematically grow under the null, regardless of how $\lambda_i$ was chosen (as long as it depends only on past events).
\end{proof}

By Ville's inequality, $\Pr_{H_0}(\sup_{i \geq 1} W_i \geq 1/\alpha) \leq \alpha$. Rejecting when wealth crosses $1/\alpha$ controls Type I error at any stopping time.

\subsection{Adaptive bidirectionality}

The adaptive e-RTe wager automatically detects both benefit and harm without pre-specifying a direction:
\begin{itemize}
    \item If $\hat{p} < 0.5$ (fewer treatment events than expected): the method bets on control events being more common, and wealth grows when treatment is indeed protective.
    \item If $\hat{p} > 0.5$ (more treatment events than expected): the method bets on treatment events being more common, and wealth grows when treatment is harmful.
\end{itemize}

No investigator input about the expected direction is needed. The adaptive wager discovers the direction from the data. The same principle also applies to adaptive e-RTb: if events become more common in treatment, events point toward treatment; if events become less common in treatment, events point toward control. This bidirectionality is a property of adaptive wagers. A design-fixed wager is instead directional unless a separate two-sided or mirrored design rule is specified.

\subsection{Signal Concentration}

A key property of e-RTe is that it can outperform the full-sample e-RTb when the baseline event rate is low. The intuition is that events \emph{concentrate} the treatment signal.

Consider a trial where the control event rate is 25\% and the treatment event rate is 20\%, yielding a 5 percentage-point absolute risk reduction (ARR). In e-RTb, the signal per patient is diluted: most patients do not experience the event, and only the 20--25\% who do carry information about differential event rates. The observed risk difference across all patients is 5 percentage points.

In e-RTe, only events are observed. The event-coin probability is:
\begin{equation}
p_{\text{alt}} = \frac{p_{\text{trt}}}{p_{\text{trt}} + p_{\text{ctrl}}} = \frac{0.20}{0.20 + 0.25} = 0.444
\end{equation}
This is an 11.2-point tilt from 0.5---more than double the 5-point ARR. The signal is concentrated because events filter out the uninformative non-events.

This advantage diminishes as the baseline event rate increases. At higher event rates, e-RTb sees more informative events per patient, and the event-coin tilt shrinks because both numerator and denominator grow:

\begin{table}[htbp]
\centering
\caption{Signal concentration: event-coin tilt versus ARR for a 5pp risk reduction.}
\begin{tabular}{@{}ccccc@{}}
\toprule
Baseline Event Rate & Treatment Event Rate & Event Coin $p_{\text{alt}}$ & Tilt from 0.5 & Tilt / ARR \\
\midrule
10\% & 5\% & 0.333 & 16.7 pp & 3.33$\times$ \\
15\% & 10\% & 0.400 & 10.0 pp & 2.00$\times$ \\
20\% & 15\% & 0.429 & 7.1 pp & 1.43$\times$ \\
25\% & 20\% & 0.444 & 5.6 pp & 1.11$\times$ \\
30\% & 25\% & 0.455 & 4.5 pp & 0.91$\times$ \\
35\% & 30\% & 0.462 & 3.8 pp & 0.77$\times$ \\
40\% & 35\% & 0.467 & 3.3 pp & 0.67$\times$ \\
\bottomrule
\end{tabular}
\label{tab:signal_concentration}
\end{table}

The analytical tilt calculation places the crossover with e-RTb near 25\% baseline event rate. The finite-sample head-to-head simulation below remains slightly e-RTe-favorable at 25\% and switches to e-RTb by 30\%. Below that region, the concentrated event-coin signal more than compensates for the smaller number of observations; above it, e-RTb's access to all patients provides the advantage.

\subsection{Trajectory Examples}

Figure~\ref{fig:erte} shows representative e-RTe trajectories. In this 500-event alternative illustration, 12 of 30 paths cross the threshold as the adaptive wager learns the event-coin imbalance.

\begin{figure}[htbp]
\centering
\includegraphics[width=0.48\textwidth]{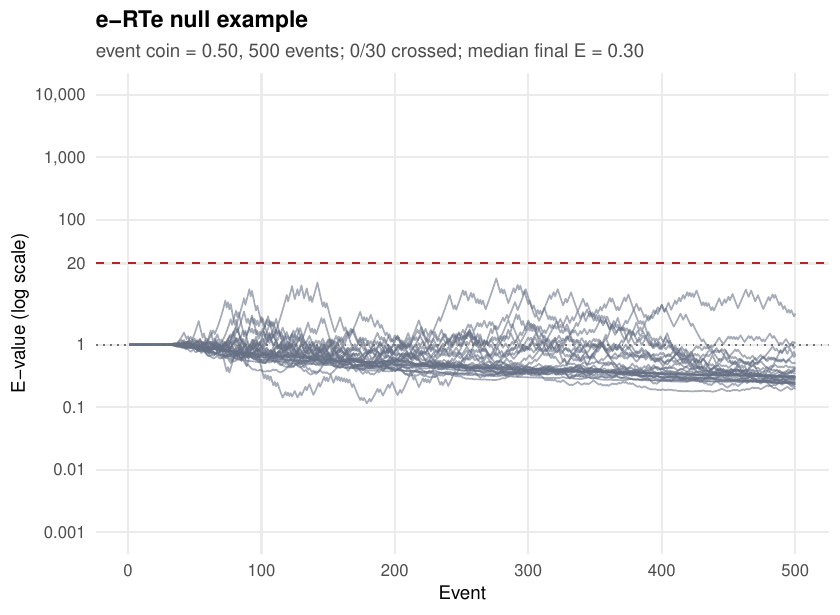}
\includegraphics[width=0.48\textwidth]{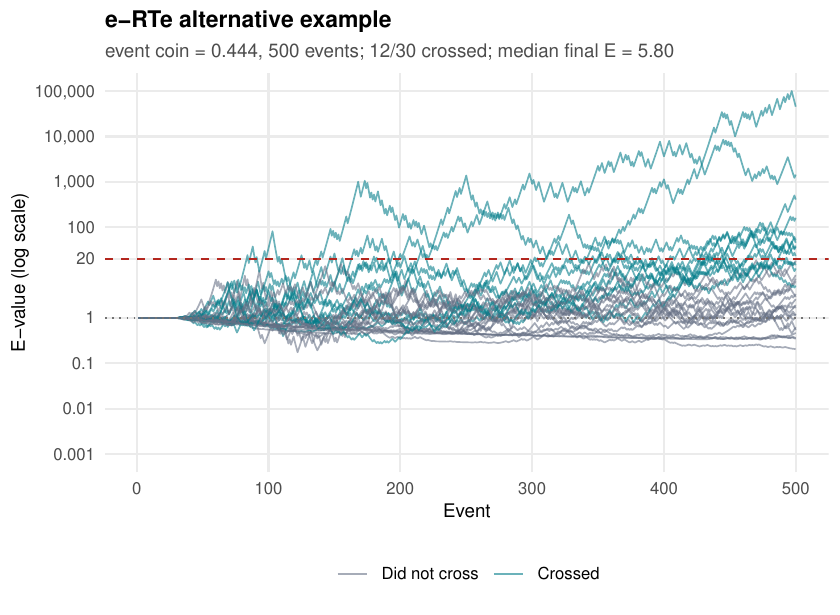}
\caption{
Trajectories of the e-RTe process (25\% baseline event rate, 5pp ARR, 500 events).
Left: under the null hypothesis (event coin $= 0.50$), wealth fluctuates randomly.
Right: under the alternative hypothesis, some paths cross as the adaptive wager learns the event-coin imbalance.
Dashed red line: rejection threshold ($1/\alpha = 20$).
}
\label{fig:erte}
\end{figure}

\subsection{Head-to-head comparison with e-RTb}

To quantify the signal concentration crossover, we ran both e-RTe and e-RTb on the \emph{same} simulated trials across a range of baseline event rates and absolute risk reductions (ARRs). For each scenario, we computed the frequentist sample size (two-proportion $z$-test, 80\% power), enrolled that many patients, and analyzed the data with both methods: e-RTb processed all patients; e-RTe processed only the event stream. We used 2{,}000 simulations per scenario. Scenarios that would imply a treatment event rate below 5\% were omitted.

\begin{table}[htbp]
\centering
\small
\caption{Head-to-head comparison: e-RTe versus e-RTb across absolute risk reductions, using the same trial data and same enrolled-patient sample size. Scenarios requiring treatment event rates below 5\% are omitted.}
\begin{tabular}{@{}ccccccccl@{}}
\toprule
ARR & Baseline & Event Coin & $N$ & Events & e-RTb Power & e-RTe Power & $\Delta$ & Winner \\
\midrule
5.0 pp & 10\% & 0.333 & 870 & 66 & 12.3\% & 8.1\% & $-4.2$pp & e-RTb \\
5.0 pp & 15\% & 0.400 & 1{,}372 & 172 & 24.8\% & 43.7\% & $+18.9$pp & e-RTe \\
5.0 pp & 20\% & 0.429 & 1{,}812 & 318 & 33.1\% & 44.7\% & $+11.6$pp & e-RTe \\
5.0 pp & 25\% & 0.444 & 2{,}188 & 493 & 40.0\% & 42.8\% & $+2.8$pp & e-RTe \\
5.0 pp & 30\% & 0.455 & 2{,}502 & 689 & 42.5\% & 37.4\% & $-5.1$pp & e-RTb \\
5.0 pp & 35\% & 0.462 & 2{,}754 & 896 & 46.0\% & 34.2\% & $-11.8$pp & e-RTb \\
5.0 pp & 40\% & 0.467 & 2{,}942 & 1{,}104 & 48.9\% & 32.1\% & $-16.8$pp & e-RTb \\
7.5 pp & 15\% & 0.333 & 556 & 63 & 19.7\% & 4.9\% & $-14.9$pp & e-RTb \\
7.5 pp & 20\% & 0.385 & 758 & 124 & 30.4\% & 35.5\% & $+5.1$pp & e-RTe \\
7.5 pp & 25\% & 0.412 & 932 & 199 & 40.5\% & 41.5\% & $+1.0$pp & e-RTe \\
7.5 pp & 30\% & 0.429 & 1{,}080 & 284 & 44.4\% & 37.0\% & $-7.4$pp & e-RTb \\
7.5 pp & 35\% & 0.440 & 1{,}198 & 375 & 49.0\% & 35.4\% & $-13.6$pp & e-RTb \\
7.5 pp & 40\% & 0.448 & 1{,}288 & 467 & 50.2\% & 33.4\% & $-16.9$pp & e-RTb \\
10.0 pp & 15\% & 0.250 & 282 & 29 & 10.8\% & 0.0\% & $-10.8$pp & e-RTb \\
10.0 pp & 20\% & 0.333 & 398 & 60 & 25.4\% & 2.8\% & $-22.6$pp & e-RTb \\
10.0 pp & 25\% & 0.375 & 500 & 100 & 33.1\% & 26.1\% & $-7.1$pp & e-RTb \\
10.0 pp & 30\% & 0.400 & 588 & 148 & 42.5\% & 33.3\% & $-9.2$pp & e-RTb \\
10.0 pp & 35\% & 0.417 & 658 & 198 & 46.2\% & 33.5\% & $-12.7$pp & e-RTb \\
10.0 pp & 40\% & 0.429 & 712 & 250 & 50.7\% & 34.2\% & $-16.4$pp & e-RTb \\
\bottomrule
\end{tabular}
\label{tab:compare_erte_binary}
\end{table}

For a 5pp ARR, the crossover occurs near 25\% baseline event rate, consistent with the analytical prediction from Table~\ref{tab:signal_concentration}. At 15--20\% baseline, e-RTe outperforms e-RTb despite seeing fewer observations, because the event-coin tilt more than compensates for the smaller event stream. Above 30\%, e-RTb's access to all patients provides an increasingly large advantage. For larger ARRs, the frequentist sample size shrinks and the event stream may become too short for e-RTe to complete its burn-in and ramp. In those scenarios, e-RTb usually dominates because the larger per-patient effect can be exploited immediately across all randomized patients.

\paragraph{Low-event constraints.} At 10\% baseline with a 5pp ARR, the frequentist sample size produces only about 66 expected events---fewer than the burn-in (30) plus ramp (50) = 80 events required for e-RTe to reach full betting strength. The process cannot fully leverage the strong event-coin tilt unless the planned event stream is long enough.

\paragraph{Larger effects.} With larger effects (7.5pp or 10pp ARR), the event-coin tilt is stronger but the planned sample size is smaller, so e-RTe may have too few events to learn before the trial ends.

\paragraph{Practical guidance.} Use e-RTe when (i) the baseline event rate is low-to-moderate, (ii) the expected ARR is modest enough that the planned sample still yields a sufficiently long event stream, and (iii) ascertainment of non-events is impractical. When baseline event rates are high or effects are large enough to make the planned trial short, e-RTb is usually more powerful.

\subsection{Burn-in, ramp, and Kelly-intensity sensitivity}

The previous comparison used the default adaptive settings: e-RTb used a 50-patient burn-in, 100-patient ramp, and half-intensity adaptive wager; e-RTe used a 30-event burn-in, 50-event ramp, and full-intensity adaptive wager. These defaults are starting points. Because e-RTe updates only at events, a fixed 30/50 event schedule may be too slow when the planned trial yields few events. Conversely, overly aggressive e-RTb betting can degrade wealth through many patient-level updates.

We therefore repeated the same-N comparison over a tuning grid. For each endpoint, we crossed fixed 10/20, fixed 30/50, default, proportional 5/10\%, and proportional 10/20\% burn-in/ramp schedules with 25\%, 50\%, 75\%, and 100\% adaptive Kelly intensity. Proportional schedules were defined on the natural update scale: planned patients for e-RTb and expected events for e-RTe. This is a sensitivity analysis rather than an optimized rule; choosing the best row after seeing trial data would not be a prespecified monitoring procedure. The resulting operating characteristics are summarized in Table~\ref{tab:erte_tuning_sensitivity}.

\begin{table}[htbp]
\centering
\scriptsize
\setlength{\tabcolsep}{3pt}
\caption{Sensitivity of adaptive e-RTb and e-RTe to burn-in/ramp schedule and Kelly intensity. Each row uses the same trial-data scenarios as Table~\ref{tab:compare_erte_binary}; values summarize 1{,}000 fixed-seed simulations. Tuned columns show the best power over fixed 10/20-event, fixed 30/50-event, default, proportional 5/10\%, and proportional 10/20\% burn-in/ramp schedules crossed with 25\%, 50\%, 75\%, and 100\% Kelly intensity.}
\begin{tabular}{@{}cccccccl@{}}
\toprule
ARR & Baseline & Events & e-RTb Default & e-RTb Tuned & e-RTe Default & e-RTe Tuned & Winner \\
\midrule
5.0 pp & 10\% & 66 & 12.2\% & 19.6\% (30/50 events, 100\%K) & 6.4\% & 48.3\% (5/10\%, 100\%K) & e-RTe \\
5.0 pp & 15\% & 172 & 25.2\% & 27.7\% (5/10\%, 75\%K) & 42.0\% & 48.3\% (5/10\%, 100\%K) & e-RTe \\
5.0 pp & 20\% & 318 & 32.9\% & 36.4\% (10/20\%, 75\%K) & 43.1\% & 43.1\% (default, 100\%K) & e-RTe \\
5.0 pp & 25\% & 493 & 40.9\% & 41.2\% (10/20\%, 75\%K) & 42.6\% & 42.6\% (default, 100\%K) & e-RTe \\
5.0 pp & 30\% & 689 & 43.8\% & 46.9\% (10/20\%, 75\%K) & 37.2\% & 38.0\% (5/10\%, 100\%K) & e-RTb \\
5.0 pp & 35\% & 896 & 47.9\% & 51.7\% (10/20\%, 75\%K) & 36.1\% & 38.8\% (5/10\%, 100\%K) & e-RTb \\
5.0 pp & 40\% & 1{,}104 & 46.4\% & 51.0\% (5/10\%, 50\%K) & 29.3\% & 30.1\% (5/10\%, 100\%K) & e-RTb \\
7.5 pp & 15\% & 63 & 20.6\% & 27.0\% (default, 75\%K) & 6.6\% & 46.6\% (10/20\%, 100\%K) & e-RTe \\
7.5 pp & 20\% & 124 & 31.6\% & 34.0\% (10/20\%, 75\%K) & 32.5\% & 46.1\% (10/20 events, 100\%K) & e-RTe \\
7.5 pp & 25\% & 199 & 37.9\% & 40.0\% (10/20\%, 75\%K) & 38.5\% & 40.2\% (10/20\%, 100\%K) & $\sim$Tied \\
7.5 pp & 30\% & 284 & 42.3\% & 44.2\% (10/20\%, 75\%K) & 38.6\% & 38.9\% (10/20\%, 100\%K) & e-RTb \\
7.5 pp & 35\% & 375 & 46.6\% & 46.6\% (default, 50\%K) & 35.3\% & 36.1\% (5/10\%, 100\%K) & e-RTb \\
7.5 pp & 40\% & 467 & 51.8\% & 53.1\% (10/20\%, 75\%K) & 35.9\% & 36.6\% (10/20\%, 100\%K) & e-RTb \\
10.0 pp & 15\% & 29 & 13.5\% & 25.5\% (default, 100\%K) & 0.0\% & 46.8\% (5/10\%, 75\%K) & e-RTe \\
10.0 pp & 20\% & 60 & 22.2\% & 31.0\% (default, 100\%K) & 2.7\% & 41.9\% (10/20\%, 100\%K) & e-RTe \\
10.0 pp & 25\% & 100 & 37.1\% & 41.4\% (default, 75\%K) & 25.2\% & 43.4\% (10/20 events, 100\%K) & e-RTe \\
10.0 pp & 30\% & 148 & 42.4\% & 42.6\% (10/20\%, 75\%K) & 37.5\% & 43.9\% (10/20 events, 100\%K) & e-RTe \\
10.0 pp & 35\% & 198 & 45.9\% & 48.2\% (10/20\%, 75\%K) & 34.6\% & 36.7\% (10/20 events, 100\%K) & e-RTb \\
10.0 pp & 40\% & 250 & 51.3\% & 52.2\% (10/20\%, 75\%K) & 33.6\% & 34.0\% (10/20\%, 100\%K) & e-RTb \\
\bottomrule
\end{tabular}
\label{tab:erte_tuning_sensitivity}
\end{table}

The sensitivity analysis confirms that e-RTe performance is partly a tuning issue. Shorter or proportional ramps can recover power when the default 30/50 event schedule consumes too much of the event stream. This is design sensitivity, not a data-adaptive tuning rule. As event streams become longer and baseline risk becomes higher, e-RTb often regains the advantage because it updates for every randomized patient and uses non-event information.

\section{Design-Calibrated Wagers for Binary and Event-Only Endpoints}

The preceding simulations used the default wager policies: adaptive half-Kelly for e-RTb and adaptive full-Kelly for e-RTe. We explicitly separate the validity engine from the wager policy. The martingale argument requires only that the wager be predictable: it may be learned from prior trial data, or it may be fixed in advance from the design alternative. This distinction parallels \citet{sokolova2026evalues}, where the e-process is calibrated to a clinically meaningful design effect.

A growth-rate-optimal (GROW) wager \citep{ramdas2025} is the value $\lambda^\star$ that maximizes the expected log-growth of the e-process under a specified design alternative:
\begin{equation}
\lambda^\star = \arg\max_\lambda \; \mathbb{E}_{\text{design}}\{\log M(\lambda)\},
\end{equation}
where $M(\lambda)$ is the one-step e-process multiplier. In the binary paired-comparison construction of \citet{sokolova2026evalues}, this wager is computed from the design alternative before monitoring and then held fixed during the e-value run. This is analogous to a frequentist design effect used for sample-size planning: it improves efficiency when the design alternative is close to the truth, but can lose power when the effect is misspecified. The design-fixed e-RT policies below adopt the same planning philosophy, while retaining the e-RT assignment-prediction construction rather than the paired-comparison construction.

For e-RTb, the design-fixed wager uses the posterior assignment probabilities implied by prespecified event rates. With 1:1 randomization and design rates $p_T$ and $p_C$,
\begin{equation}
\lambda_{\text{event}} =
\Pr(T = 1 \mid Y = 1)
= \frac{p_T}{p_T + p_C},
\qquad
\lambda_{\text{non-event}} =
\Pr(T = 1 \mid Y = 0)
= \frac{1 - p_T}{(1 - p_T) + (1 - p_C)}.
\end{equation}
For e-RTe, the design-fixed wager is the corresponding event-coin probability, $\Pr(T = 1 \mid \text{event}) = p_T/(p_T + p_C)$. Thus e-RTb fixes both an event and non-event wager, whereas e-RTe fixes only the event wager.

We ran 5{,}000 simulations per scenario using a fixed seed. Sample sizes were anchored to the usual fixed-sample two-proportion calculation in R (\texttt{power.prop.test}, 80\% power, $\alpha = 0.05$). Both e-RTb and e-RTe were evaluated at the same enrolled-patient $N$; no event-only inflation was used in this comparison. Under the null, we reused the design $N$ for 5pp and 10pp ARR alternatives. Under the alternative, we compared adaptive wagering against fixed wagers that underestimated, matched, or overestimated the true ARR. An oracle full-Kelly row is included only as a simulation benchmark. Type I error and power are reported in Tables~\ref{tab:wager_policy_type1} and~\ref{tab:wager_policy_v8}.

\begin{table}[htbp]
\centering
\caption{Type I error for adaptive and design-fixed wager policies at the same enrolled-patient sample sizes. Each row summarizes 5{,}000 simulated null trials with $p_C = 0.40$ and $\alpha = 0.05$. Sample sizes were obtained from the usual fixed-sample two-proportion calculation with 80\% power; no event-only inflation was used for e-RTe in this comparison.}
\label{tab:wager_policy_type1}
\begin{tabular}{@{}lllrrrr@{}}
\toprule
Endpoint & Scenario & Policy & Wager ARR & $N$ & Events & Type I \\
\midrule
e-RTb & Null 5.0pp design & Adaptive & -- & 2{,}942 & -- & 3.5\% \\
e-RTb & Null 5.0pp design & Fixed 5pp & 5.0pp & 2{,}942 & -- & 3.7\% \\
e-RTb & Null 5.0pp design & Fixed 10pp & 10.0pp & 2{,}942 & -- & 4.5\% \\
e-RTb & Null 10.0pp design & Adaptive & -- & 712 & -- & 2.0\% \\
e-RTb & Null 10.0pp design & Fixed 5pp & 5.0pp & 712 & -- & 0.3\% \\
e-RTb & Null 10.0pp design & Fixed 10pp & 10.0pp & 712 & -- & 3.3\% \\
e-RTe & Null 5.0pp design & Adaptive & -- & 2{,}942 & 1{,}177 & 3.2\% \\
e-RTe & Null 5.0pp design & Fixed 5pp & 5.0pp & 2{,}942 & 1{,}177 & 2.7\% \\
e-RTe & Null 5.0pp design & Fixed 10pp & 10.0pp & 2{,}942 & 1{,}177 & 4.8\% \\
e-RTe & Null 10.0pp design & Adaptive & -- & 712 & 285 & 1.7\% \\
e-RTe & Null 10.0pp design & Fixed 5pp & 5.0pp & 712 & 285 & 0.1\% \\
e-RTe & Null 10.0pp design & Fixed 10pp & 10.0pp & 712 & 285 & 2.0\% \\
\bottomrule
\end{tabular}
\end{table}

\begin{table}[htbp]
\centering
\caption{Power for adaptive and design-fixed wager policies at the same enrolled-patient sample sizes. Each row summarizes 5{,}000 simulated trials with $p_C = 0.40$ and $\alpha = 0.05$. Fixed policies use full-Kelly wagers calibrated to the listed wager ARR; adaptive e-RTb uses half-Kelly and adaptive e-RTe uses full-Kelly.}
\label{tab:wager_policy_v8}
\begin{tabular}{@{}lllrrrrr@{}}
\toprule
Endpoint & Scenario & Policy & Wager ARR & $N$ & Events & Power & Median crossing \\
\midrule
e-RTb & True 5.0pp & Adaptive & -- & 2{,}942 & -- & 49.3\% & 1{,}452 \\
e-RTb & True 5.0pp & Fixed under & 2.5pp & 2{,}942 & -- & 53.7\% & 2{,}154 \\
e-RTb & True 5.0pp & Fixed matched & 5.0pp & 2{,}942 & -- & 75.0\% & 1{,}438 \\
e-RTb & True 5.0pp & Fixed over & 10.0pp & 2{,}942 & -- & 55.7\% & 824 \\
e-RTb & True 5.0pp & Oracle & 5.0pp & 2{,}942 & -- & 73.9\% & 1{,}453 \\
e-RTb & True 10.0pp & Adaptive & -- & 712 & -- & 50.5\% & 396 \\
e-RTb & True 10.0pp & Fixed under & 5.0pp & 712 & -- & 41.9\% & 565 \\
e-RTb & True 10.0pp & Fixed matched & 10.0pp & 712 & -- & 71.3\% & 404 \\
e-RTb & True 10.0pp & Fixed over & 15.0pp & 712 & -- & 67.1\% & 327 \\
e-RTb & True 10.0pp & Oracle & 10.0pp & 712 & -- & 71.4\% & 409 \\
e-RTe & True 5.0pp & Adaptive & -- & 2{,}942 & 1{,}104 & 31.5\% & 556 \\
e-RTe & True 5.0pp & Fixed under & 2.5pp & 2{,}942 & 1{,}104 & 14.2\% & 936 \\
e-RTe & True 5.0pp & Fixed matched & 5.0pp & 2{,}942 & 1{,}104 & 51.2\% & 695 \\
e-RTe & True 5.0pp & Fixed over & 10.0pp & 2{,}942 & 1{,}104 & 46.9\% & 404 \\
e-RTe & True 5.0pp & Oracle & 5.0pp & 2{,}942 & 1{,}104 & 52.2\% & 686 \\
e-RTe & True 10.0pp & Adaptive & -- & 712 & 250 & 33.8\% & 158 \\
e-RTe & True 10.0pp & Fixed under & 5.0pp & 712 & 250 & 5.7\% & 227 \\
e-RTe & True 10.0pp & Fixed matched & 10.0pp & 712 & 250 & 43.2\% & 182 \\
e-RTe & True 10.0pp & Fixed over & 15.0pp & 712 & 250 & 50.6\% & 145 \\
e-RTe & True 10.0pp & Oracle & 10.0pp & 712 & 250 & 42.4\% & 181 \\
\bottomrule
\end{tabular}
\end{table}

The adaptive null rows in Table~\ref{tab:wager_policy_type1} are now directly comparable in precision to the 5{,}000-replicate adaptive e-RTb baseline in Table~\ref{tab:simulations}. At the 10pp design sample size, the adaptive e-RTb null estimate is 2.0\% here and 2.1\% in the baseline run, supporting the interpretation that the lower Type I error in the shorter 10pp-design trial reflects the operating scenario rather than a separate setup difference.

Design-fixed full-Kelly wagers increased power when the design effect was close to the truth while retaining Type I error control in these simulations. Misspecification mattered: underestimating the effect was usually conservative, while overestimating it produced earlier crossings among successful trials but sometimes lower overall power, especially for e-RTe at a true 5pp ARR. Thus design-fixed wagers are useful as optional efficiency tools rather than replacements for adaptive effect-size-agnostic monitoring. Figures~\ref{fig:wager_policy_type1}--\ref{fig:wager_policy_traces} show the corresponding operating characteristics and representative wealth paths.

\begin{figure}[htbp]
\centering
\includegraphics[width=0.9\textwidth]{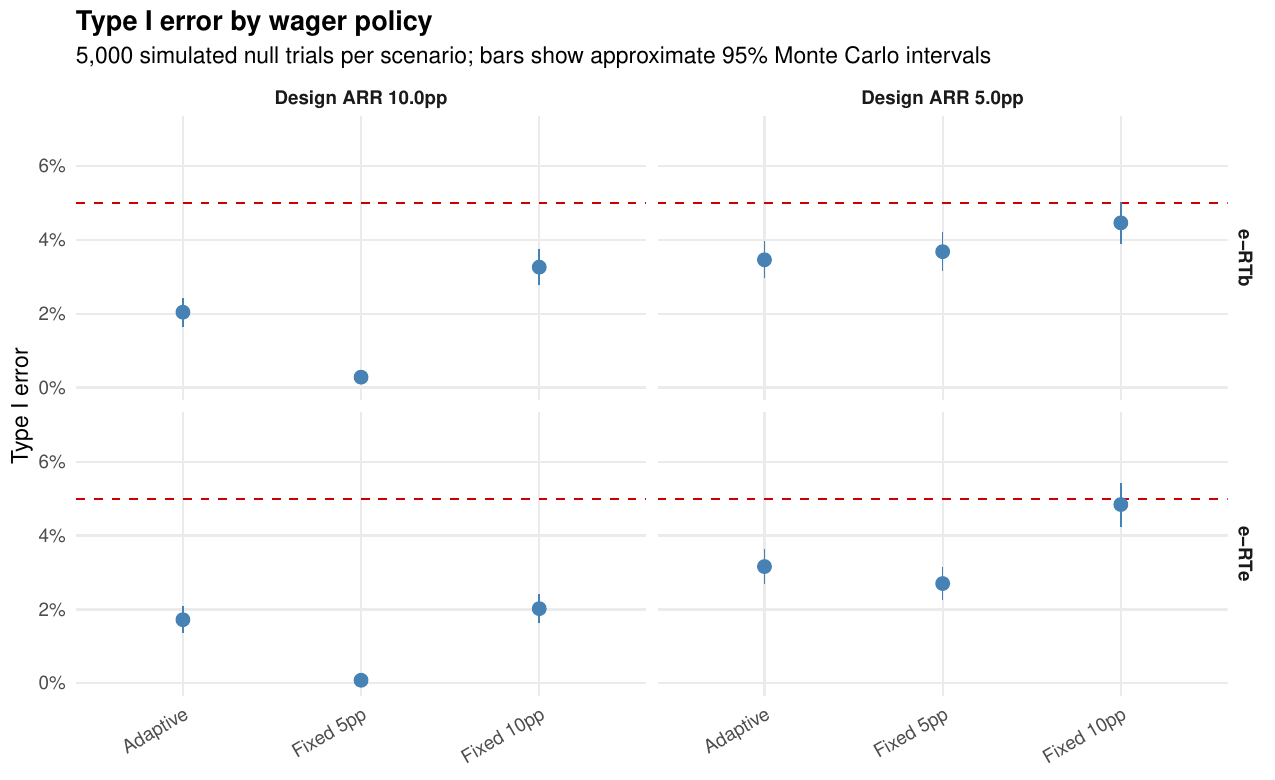}
\caption{Type I error for adaptive and design-fixed wager policies. Each point summarizes 5{,}000 null simulations. The dashed red line is the nominal $\alpha = 0.05$ threshold.}
\label{fig:wager_policy_type1}
\end{figure}

\begin{figure}[htbp]
\centering
\includegraphics[width=0.9\textwidth]{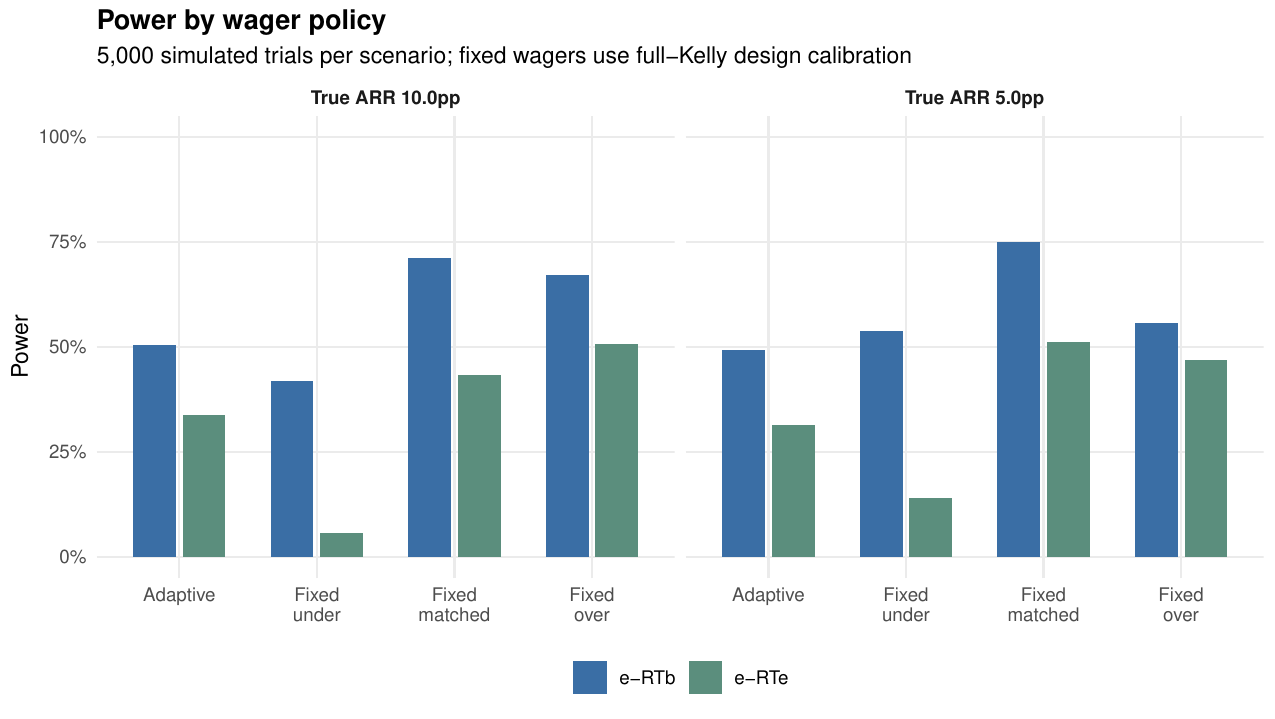}
\caption{Power for adaptive and design-fixed wager policies. Fixed matched wagers are calibrated to the true ARR; underestimated and overestimated wagers deliberately misspecify the design effect.}
\label{fig:wager_policy_power}
\end{figure}

\begin{figure}[htbp]
\centering
\includegraphics[width=0.98\textwidth]{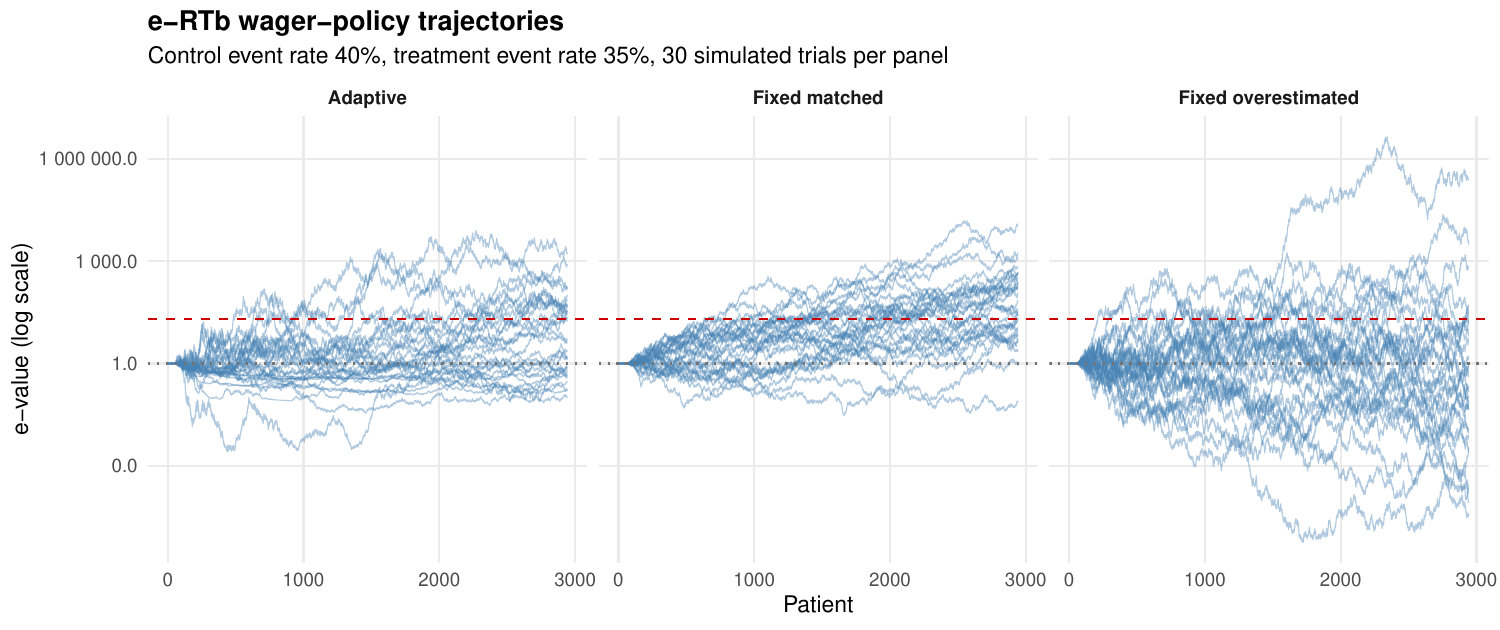}
\includegraphics[width=0.98\textwidth]{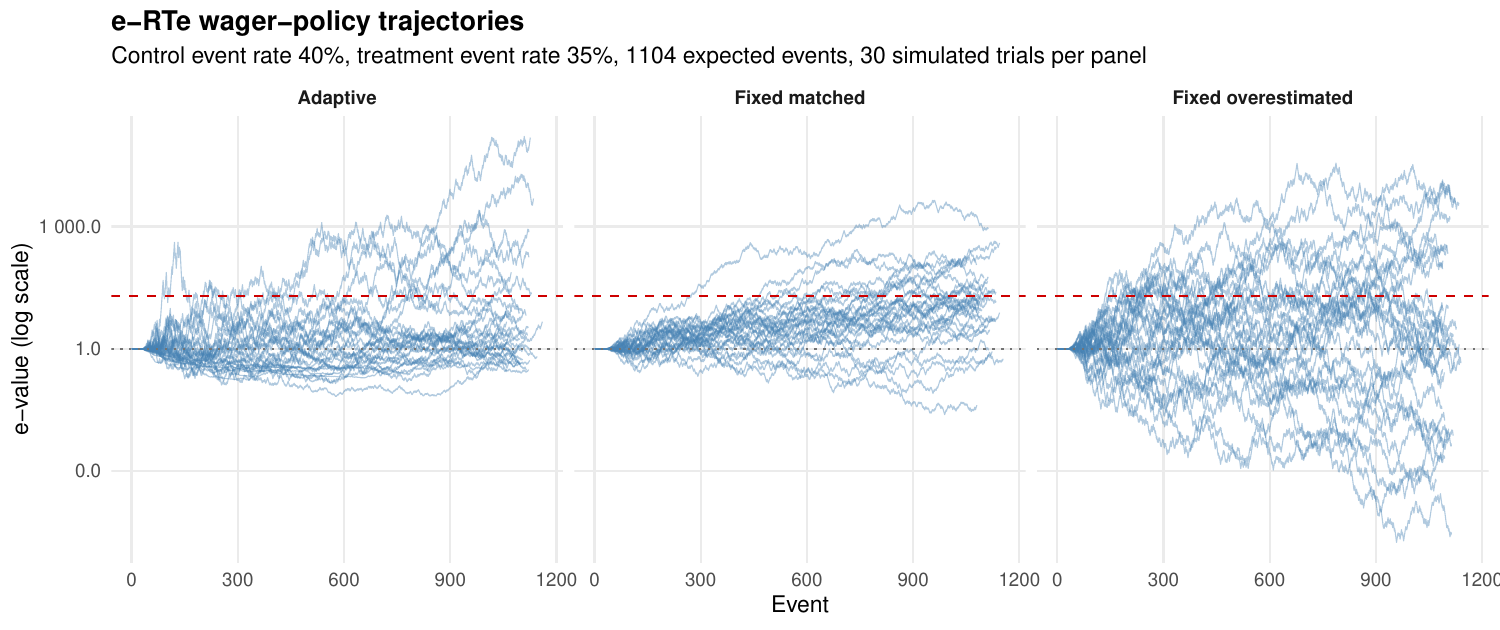}
\caption{Example wealth trajectories under a true 5pp ARR with control event rate 40\%. Top: e-RTb using all enrolled patients. Bottom: e-RTe using only events. Each panel shows 30 simulated trials.}
\label{fig:wager_policy_traces}
\end{figure}

\section{Effect Estimates at Crossing and Type M Error}

Because e-RT methods may stop early, the apparent treatment effect at the first threshold crossing is expected to be inflated. Following the design-analysis terminology of \citet{gelman2014beyond}, this is a form of Type M error: the magnitude of the observed effect among trials that cross may exceed the true effect because crossings are enriched for favorable random fluctuations. This does not invalidate the e-value---the e-process itself remains anytime-valid---but it matters for clinical interpretation. A trial stopped at an e-value threshold should not interpret the naive effect estimate at crossing as an unbiased estimate of the final treatment effect.

We therefore summarized the apparent effect at first crossing among simulated trials that crossed. For both e-RTb and e-RTe, the displayed effect scale is the apparent absolute risk reduction using all randomized patients observed by the time of first crossing. For e-RTe, this is a diagnostic snapshot rather than information used by the event-only e-process: the e-RTe wealth process itself still sees only event arm labels, but when denominators are available at the crossing time, the clinical effect size can be estimated on the same scale as e-RTb. If denominators are not available operationally, only the native event-coin diagnostic can be reported. Table~\ref{tab:type_m_crossing} reports the binary and event-only crossing diagnostics; analogous scale-specific Type M diagnostics are reported below for e-RTc.

\begin{table}[htbp]
\centering
\small
\caption{Type M error at first e-process crossing for e-RTb and e-RTe. Crossing and final effects are absolute risk reductions in percentage points. For e-RTe, the e-process itself remains event-only; the absolute risk reduction is a full-data diagnostic computed from all randomized patients observed by the crossing time. Type M is computed among trials that crossed.}
\begin{tabular}{@{}lllrrrrr@{}}
\toprule
True & Endpoint & Policy & Crossing & Final & Median M & Q75 & Q90 \\
\midrule
5.0pp & e-RTb & Adaptive & 7.92 & 5.00 & 1.58 & 2.12 & 2.95 \\
5.0pp & e-RTb & Fixed matched & 6.52 & 5.01 & 1.30 & 1.69 & 2.15 \\
5.0pp & e-RTe & Adaptive & 8.03 & 4.99 & 1.61 & 2.03 & 2.66 \\
5.0pp & e-RTe & Fixed matched & 6.73 & 4.99 & 1.35 & 1.64 & 2.00 \\
10.0pp & e-RTb & Adaptive & 14.68 & 10.06 & 1.47 & 1.79 & 2.20 \\
10.0pp & e-RTb & Fixed matched & 12.74 & 10.03 & 1.27 & 1.57 & 1.89 \\
10.0pp & e-RTe & Adaptive & 14.63 & 10.09 & 1.46 & 1.76 & 2.11 \\
10.0pp & e-RTe & Fixed matched & 13.15 & 10.02 & 1.32 & 1.58 & 1.83 \\
\bottomrule
\end{tabular}
\label{tab:type_m_crossing}
\end{table}

The inflation is clinically meaningful. In the 5pp ARR scenario, adaptive e-RTb crossings had a median apparent ARR of 7.92 percentage points, with a median Type M ratio of 1.58. Adaptive e-RTe crossings had a similar full-data snapshot ARR of 8.03 percentage points, with a median Type M ratio of 1.61. Fixed matched wagers reduced this inflation for both methods, with median ratios of 1.30 for e-RTb and 1.35 for e-RTe. These results reinforce that e-RT crossings should be interpreted as valid evidence for a treatment difference, not as unbiased estimates of its magnitude.

\subsection{Reporting after a crossing}

At a Data and Safety Monitoring Board (DSMB) review, the e-value should carry the inferential claim and conventional summaries should be labeled as descriptive. A compact crossing report should state the prespecified threshold, the observed e-value, the variant and wager policy, the update at which crossing occurred, the apparent clinical effect at crossing, and the planned final analysis. For example: ``The adaptive e-RTb process crossed the prespecified threshold of 20 at patient 1{,}450, with $W=25.4$. At that time the apparent ARR was 7.8 percentage points. Because this estimate is selected at the first crossing, it is descriptive and may overstate the final treatment effect; the final effect estimate and confidence interval will follow the prespecified primary analysis.'' The protocol should specify whether crossing triggers stopping, DSMB review, or continued monitoring.

\section{Continuous Outcomes}

Continuous endpoints use the same assignment-prediction game as e-RTb. The outcome $Y_i$ is observed before using the assignment $T_i$ in the wealth update; under the null, $Y_i$ does not help predict $T_i$. We call this extension e-RTc.

\subsection{Setup}

At enrollment $i$, treatment assignment is $T_i \in \{0,1\}$ with known allocation probability $p=P(T_i=1)$, and the outcome is $Y_i \in \mathbb{R}$. Under the null hypothesis of no treatment effect,
\begin{equation}
Y_i \perp T_i.
\end{equation}
No distributional assumption on $Y_i$ is needed for validity. Continuous outcomes change only the wager: instead of betting from an event/no-event indicator, e-RTc bets from how unusual the observed value is relative to previous outcomes.

\subsection{Betting strategy for continuous outcomes}

The default adaptive e-RTc wager uses robust standardization and a coarse running direction estimate. Using previous outcomes $Y_1, \dots, Y_{i-1}$, compute:
\begin{align}
m_{i-1} &= \text{median}(Y_1, \dots, Y_{i-1}), \\
s_{i-1} &= \text{MAD}(Y_1, \dots, Y_{i-1}) ,
\end{align}
where MAD is the median absolute deviation,
\begin{equation}
\mathrm{MAD}(Y_1,\ldots,Y_{i-1})
= \operatorname{median}_{1 \leq \ell < i}\left|Y_\ell - m_{i-1}\right|.
\end{equation}
The median and MAD are robust to outliers and skewness. If $s_{i-1}$ is zero or not finite, we set $s_{i-1} = 1$ to avoid degeneracy.

For the new patient, form the standardized residual
\begin{equation}
r_i = \frac{Y_i - m_{i-1}}{s_{i-1}}.
\end{equation}
Then shrink it into $(-1,1)$:
$$
g_i = \frac{r_i}{1 + |r_i|}.
$$
For moderate values $g_i \approx r_i$, while very large residuals are shrunk toward $+1$ or $-1$. This prevents a single extreme observation from forcing an almost all-in bet.

Betting strength is ramped over time:
\begin{equation}
c_i = \min\left\{1, \max\left(0, \frac{i - \text{burn-in}}{\text{ramp}}\right)\right\},
\end{equation}
where \texttt{burn-in} is the number of initial patients with little or no betting and \texttt{ramp} controls the transition to full strength. The default adaptive rule also estimates direction from previous arm means:
\begin{equation}
q_{i-1} = \operatorname{sign}\left(\bar{Y}_{\text{trt},i-1} - \bar{Y}_{\text{ctrl},i-1}\right),
\end{equation}
with $q_{i-1}=0$ if either arm has no prior observations. This makes adaptive e-RTc effect-size agnostic: the running data choose the direction of the bet, but not a prespecified effect magnitude. The magnitude of the wager still depends on the current observation-level score $g_i$ and the ramped betting strength.

The raw betting fraction is
\begin{equation}
\lambda_i^{\mathrm{raw}} = p + c_i \cdot c_{\max} \cdot g_i \cdot q_{i-1},
\end{equation}
then clamped to $(0,1)$ to ensure nonnegative multipliers. It is predictable because it depends only on past outcomes and the new $Y_i$, not on $T_i$. If $Y_i$ is extreme in the direction that past data associate with treatment, $\lambda_i$ moves away from $p$ toward a confident treatment bet. If $Y_i$ is typical, or if the direction has not yet been learned, $\lambda_i$ remains near neutral.

\subsection{Parametric design wager for continuous outcomes}

The adaptive wager above is intentionally agnostic to the effect size. For trials with a credible design alternative, e-RTc can also use a parametric design wager. Suppose the protocol specifies a normal-shift working model,
\[
Y \mid T=0 \sim N(\mu_C^\ast, {\sigma^\ast}^2),
\qquad
Y \mid T=1 \sim N(\mu_T^\ast, {\sigma^\ast}^2).
\]
After observing the new outcome $Y_i$, but before revealing or using $T_i$ in the betting update, the design wager is
\[
\lambda_i^\ast
= \Pr_{\text{design}}(T_i=1 \mid Y_i)
=
\frac{p f_1^\ast(Y_i)}
{p f_1^\ast(Y_i) + (1-p) f_0^\ast(Y_i)},
\]
where $f_1^\ast$ and $f_0^\ast$ are the design densities for treatment and control. As before, we may ramp from the neutral wager $p$ toward $\lambda_i^\ast$ during early enrollment.

This design wager is parametric by construction. The normal-shift model is not needed for validity; it only determines how aggressively the e-process bets. Under the null, treatment assignment remains randomized and independent of $Y_i$, so any predictable $\lambda_i(Y_i)$ gives an expected wealth multiplier of 1. Misspecifying $\mu_C^\ast$, $\mu_T^\ast$, or $\sigma^\ast$ can reduce power or increase Type M error at crossing, but it does not break the martingale guarantee.

\subsection{Wealth update}

Treatment is still randomized with probability $p$ of intervention. After we choose $\lambda_i$, the wealth updates exactly as in the binary approach:
\begin{equation}
W_i = W_{i-1} \times
\begin{cases}
\lambda_i / p & \text{if } T_i = 1 \ (\text{intervention}) \\
(1 - \lambda_i) / (1-p) & \text{if } T_i = 0 \ (\text{control})
\end{cases}
\end{equation}
with $W_0 = 1$.

The binary and continuous variants share this update; they differ only in how $\lambda_i$ is chosen.

\subsection{Worked intuition}

Imagine a 1:1 randomized trial ($p=0.5$) where the outcome is ventilator-free days, and higher is better. Suppose that after 100 patients, the median and MAD of $Y$ are roughly stable, and the treatment arm has had better outcomes so far, so $q_{100}=1$. Patient 101 has an unusually high number of ventilator-free days compared with this distribution. The standardized residual $r_{101}$ is positive and large, so $g_{101} \approx 0.8$ and, after burn-in, $c_{101} \approx 1$. With $c_{\max} = 0.6$, the raw wager is $\lambda_{101}^{\mathrm{raw}} \approx p + 0.6 \times 0.8 \times 1 = 0.98$: we strongly bet that this patient was in the intervention arm. If they indeed were, wealth increases by roughly a factor of $0.98/0.5 \approx 2$ for this one patient. If not, wealth shrinks by $(1-0.98)/0.5 \approx 0.04$.

Under the null, high values like this are just as likely in control as in intervention: we win and lose in balance, and wealth does not grow on average. Under a true benefit, such favorable outliers cluster in the intervention arm, and the bets pay off more often than not.

\subsection{Validity}

The key point is that validity does not depend on the choice of median, MAD, or the specific transformation $g_i$. It depends only on the fact that:
\begin{enumerate}
    \item $\lambda_i$ is chosen \emph{before} observing $T_i$ and depends only on past data and $Y_i$;
    \item under the null, $T_i$ is independent of $Y_i$ with $\mathbb{P}(T_i = 1) = p$.
\end{enumerate}

\begin{theorem}
Under the null hypothesis of no treatment effect, the e-process wealth process $(W_i)$ is a test martingale: for all $i$,
\[
\mathbb{E}[W_i \mid \mathcal{F}_{i-1}] = W_{i-1},
\]
where $\mathcal{F}_{i-1}$ is the sigma-field generated by all observations up to step $i-1$.
\end{theorem}

\begin{proof}
Condition on $\mathcal{F}_{i-1}$ and $Y_i$. The bet $\lambda_i$ is now fixed. Under the null, $T_i$ is independent of $Y_i$ and
\[
\mathbb{P}(T_i = 1 \mid \mathcal{F}_{i-1}, Y_i) = p, \quad
\mathbb{P}(T_i = 0 \mid \mathcal{F}_{i-1}, Y_i) = 1-p.
\]
The conditional expectation of the wealth multiplier is:
\begin{align}
\mathbb{E}\!\left[\frac{W_i}{W_{i-1}} \,\middle|\, \mathcal{F}_{i-1}, Y_i\right]
&= p \cdot \frac{\lambda_i}{p} + (1-p) \cdot \frac{1 - \lambda_i}{1-p} \\
&= \lambda_i + (1 - \lambda_i) \\
&= 1.
\end{align}
Thus $\mathbb{E}[W_i \mid \mathcal{F}_{i-1}, Y_i] = W_{i-1}$, and taking expectations over $Y_i$ yields $\mathbb{E}[W_i \mid \mathcal{F}_{i-1}] = W_{i-1}$. This shows that $(W_i)$ is a martingale with unit expectation under the null.
\end{proof}

As in the binary case, Ville's inequality implies that for any stopping time $\tau$,
\[
\mathbb{P}(W_\tau \geq 1/\alpha) \leq \alpha,
\]
so rejecting the null when $W_\tau \geq 1/\alpha$ controls Type I error at level $\alpha$, regardless of the stopping rule.

\subsection{Simulation overview}

We evaluated e-RTc using the same design philosophy as the binary and event-only simulations. For a given standardized effect size (Cohen's $d$), we first computed the fixed-sample size required for a two-sample $t$-test with 80\% power at $\alpha = 0.05$. We then simulated 1{,}000 trials per scenario with 1:1 randomization and normally distributed outcomes with common standard deviation 1. The adaptive e-RTc used the V7 sign-direction wager with burn-in $=20$, ramp $=50$, and $c_{\max}=0.6$. The design e-RTc used the normal-shift design wager above, with matched, underestimated, and overestimated design effects. Table~\ref{tab:ertc_design_type1} reports the null simulations, and Table~\ref{tab:ertc_design_power_type_m} reports power and Type M behavior under alternatives.

\begin{table}[htbp]
\centering
\caption{Type I error for e-RTc adaptive and parametric design wagers. Each row summarizes 1{,}000 simulated null trials. Sample sizes use the usual fixed-sample two-sample $t$-test with 80\% power and $\alpha = 0.05$.}
\label{tab:ertc_design_type1}
\begin{tabular}{@{}llrrr@{}}
\toprule
Design $d$ & Policy & $N$ & Type I & Median final e-value \\
\midrule
0.20 & Adaptive & 788 & 3.8\% & 0.00 \\
0.20 & Design & 788 & 3.0\% & 0.03 \\
0.40 & Adaptive & 200 & 4.3\% & 0.00 \\
0.40 & Design & 200 & 2.9\% & 0.05 \\
0.60 & Adaptive & 90 & 4.0\% & 0.00 \\
0.60 & Design & 90 & 1.2\% & 0.21 \\
\bottomrule
\end{tabular}
\end{table}

\begin{table}[htbp]
\centering
\caption{Power and Type M error for e-RTc adaptive and parametric design wagers. The effect scale is Cohen's $d$; Type M is computed among trials that crossed.}
\label{tab:ertc_design_power_type_m}
\begin{tabular}{@{}lllrrrrr@{}}
\toprule
True $d$ & Policy & Wager $d$ & $N$ & Power & Median crossing & Crossing $d$ & Median M \\
\midrule
0.20 & Adaptive & -- & 788 & 9.8\% & 64 & 0.60 & 2.98 \\
0.20 & Design under & 0.10 & 788 & 52.1\% & 587 & 0.27 & 1.33 \\
0.20 & Design matched & 0.20 & 788 & 73.4\% & 401 & 0.26 & 1.28 \\
0.20 & Design over & 0.40 & 788 & 59.2\% & 252 & 0.32 & 1.58 \\
0.40 & Adaptive & -- & 200 & 31.6\% & 67 & 0.67 & 1.66 \\
0.40 & Design under & 0.20 & 200 & 34.4\% & 168 & 0.56 & 1.39 \\
0.40 & Design matched & 0.40 & 200 & 66.6\% & 130 & 0.51 & 1.29 \\
0.40 & Design over & 0.80 & 200 & 61.1\% & 93 & 0.55 & 1.37 \\
0.60 & Adaptive & -- & 90 & 53.8\% & 62 & 0.77 & 1.28 \\
0.60 & Design under & 0.30 & 90 & 5.7\% & 85 & 0.96 & 1.60 \\
0.60 & Design matched & 0.60 & 90 & 44.7\% & 77 & 0.78 & 1.29 \\
0.60 & Design over & 1.20 & 90 & 55.5\% & 67 & 0.75 & 1.25 \\
\bottomrule
\end{tabular}
\end{table}

Type I error remained controlled in these simulations. The adaptive sign-direction wager is conservative for small effects: at true $d=0.20$, it crossed in only 9.8\% of trials and did so very early when it crossed, producing a high median Type M ratio of 2.98. The parametric design wager improved power substantially when the design effect was close to the truth: at $d=0.20$, matched design wagering increased power to 73.4\% and reduced median Type M to 1.28. The tradeoff is visible under misspecification. Overestimating the design effect generally caused earlier crossings and more Type M inflation, while underestimating the effect delayed crossings and sometimes reduced power. These results support design-calibrated e-RTc as an optional efficiency mode, not as a replacement for the adaptive effect-size-agnostic monitor.

\begin{figure}[htbp]
\centering
\includegraphics[width=0.48\textwidth]{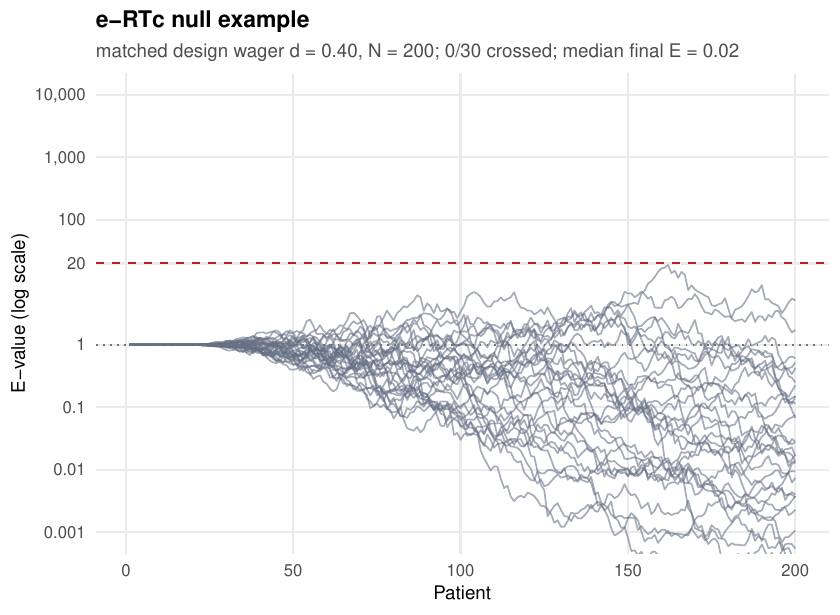}
\includegraphics[width=0.48\textwidth]{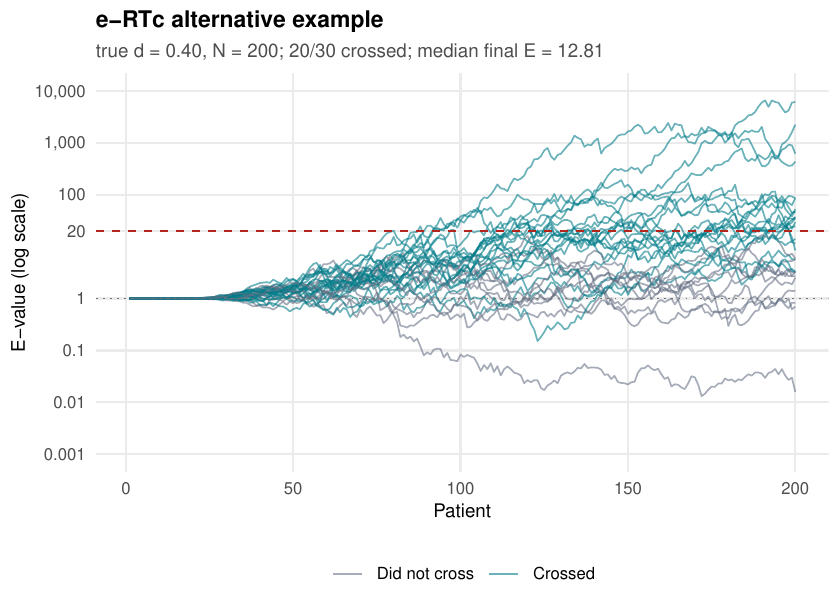}
\caption{
Trajectories of the continuous randomization e-process using the matched normal-shift design wager for a trial designed to detect a standardized mean difference of $d = 0.40$ with 80\% power. Left: trajectories under the null hypothesis ($d = 0$), where wealth usually drifts below 1 and does not cross in the representative panel. Right: trajectories under the alternative hypothesis ($d = 0.40$), where wealth grows systematically and many paths cross the rejection threshold before the planned sample size is reached.
}
\label{fig:eRTC_null_d04_80}
\end{figure}

\section{Betting Strategy Design Across Endpoints}

The active e-RT variants use different betting strategies. This is not a validity distinction: all variants rely on the same predictable-martingale argument. The difference is an efficiency distinction. A wager that is reasonable for one endpoint may be inefficient for another because the wealth product is updated on a different scale.

\subsection{Update Density and Over-Betting}

The wealth process is a product of multipliers. When a wager is too aggressive for the true effect, unfavorable multipliers are compounded repeatedly. The practical cost depends on update density: event-only methods update only when events occur, whereas binary and continuous endpoints update for every patient. This does not change Type I error control, but it can strongly affect power and the distribution of final e-values.

\subsection{Sparse-Update Endpoint: e-RTe}

The event-only e-RTe updates sparsely: the update stream is the sequence of observed events. This permits a more aggressive default than e-RTb, and the event-coin parameter is naturally bounded between 0 and 1. The tuning simulations in Table~\ref{tab:erte_tuning_sensitivity} show, however, that sparse updating creates its own design problem. If the burn-in and ramp consume much of the planned event stream, e-RTe may not begin meaningful betting soon enough. Thus e-RTe can use full-Kelly intensity, but its burn-in and ramp should be chosen with the expected number of events in mind.

\subsection{Dense-Update Endpoints: e-RTb and e-RTc}

For e-RTb, wealth updates for every randomized patient. This makes naive over-betting costly. Table~\ref{tab:wage_binary} isolates this issue using fixed deviations of the assignment-prediction wager from 0.5 under a true 5pp ARR: for events the wager is $0.5-m$ and for non-events it is $0.5+m$. Here $m$ is a betting magnitude, not an ARR. A 0.05 deviation crossed in 57.3\% of simulations, whereas twofold and threefold larger deviations crossed in only 24.3\% and 13.2\% of simulations, with median final e-values effectively zero. This stress test is intentionally narrow: it does not say that all fixed wagers are poor. Rather, it shows why the adaptive binary default uses half-Kelly shrinkage and why design-fixed binary wagers should be calibrated to clinically plausible effects.

For e-RTc, the update density is also every patient, and the wager additionally varies with an observation-level score. The adaptive sign-direction rule therefore remains conservative, especially for small effects. Table~\ref{tab:ertc_design_power_type_m} shows that at true $d=0.20$, adaptive e-RTc crossed in only 9.8\% of simulations and had substantial Type M inflation among crossings. A matched normal-shift design wager improved power to 73.4\% and reduced median Type M to 1.28, but this efficiency depends on the design model being credible. Continuous endpoints therefore emphasize the same separation seen elsewhere: validity is supplied by randomization and predictability; efficiency depends on the wager.

\subsection{Design Principle}

Table~\ref{tab:wage_summary} summarizes the practical hierarchy across active variants. Sparse-update methods tolerate more aggressive wagers because there are fewer opportunities for over-betting to compound. Dense-update methods require more shrinkage or stronger design justification. This is a statement about power and stability, not validity: any predictable wager remains valid under the null, but poorly calibrated wagers may spend their wealth in the wrong places.

The stress table below gives the binary example underlying this principle.

\begin{table}[htbp]
\centering
\caption{Binary: cost of naive over-betting (true ARR = 5\%, $N = 2{,}942$). Fixed wager magnitudes are deviations from 0.5 in the assignment-prediction wager, not ARR design effects. Rows summarize 1{,}000 fixed-seed simulated trials.}
\begin{tabular}{@{}cccc@{}}
\toprule
Fixed wager magnitude & Scale vs 0.05 wager & Power & Median Final E-value \\
\midrule
0.05 & $1\times$ & 57.3\% & 1.31 \\
0.10 & $2\times$ & 24.3\% & $\approx 0$ \\
0.15 & $3\times$ & 13.2\% & $\approx 0$ \\
0.20 & $4\times$ & 11.6\% & $\approx 0$ \\
\bottomrule
\end{tabular}
\label{tab:wage_binary}
\end{table}

The summary table below should be read as design guidance rather than a theorem.

\begin{table}[htbp]
\centering
\small
\caption{Summary: wager-policy design considerations by e-RT variant.}
\begin{tabular}{@{}lccc@{}}
\toprule
Property & e-RTe & e-RTb & e-RTc \\
\midrule
Update unit & Event & Patient & Patient \\
Default policy & Full Kelly & Half Kelly & Sign direction \\
Design-wager role & Event-rate planning & ARR planning & Normal-shift model \\
Over-betting sensitivity & Moderate & High & High \\
Main tuning issue & Event-stream length & Kelly shrinkage & Scale/model choice \\
\bottomrule
\end{tabular}
\label{tab:wage_summary}
\end{table}

\subsection{Protocol prespecification checklist}

For protocol use, the monitoring plan should make the wager policy as explicit as the stopping boundary. Table~\ref{tab:prespec_checklist} gives a compact checklist.

\begin{table}[htbp]
\centering
\small
\caption{Items to prespecify before using an e-RT monitor.}
\label{tab:prespec_checklist}
\begin{tabular}{p{0.28\textwidth}p{0.62\textwidth}}
\toprule
Item & Prespecification target \\
\midrule
Endpoint and variant & State whether monitoring uses e-RTb, e-RTe, or e-RTc, and define the update unit. \\
Validity threshold & Choose $\alpha$ and the e-value threshold $1/\alpha$; for $\alpha=0.05$, the threshold is 20. \\
Wager policy & Specify adaptive, fixed, design-calibrated, or other predictable policy; record all tuning constants. \\
Burn-in and ramp & Define burn-in and ramp on the natural update scale: patients or events. \\
Design alternative & If using a design wager, state the ARR, event-coin tilt, Cohen's $d$, and any working model parameters. \\
Stopping action & State whether crossing triggers automatic stopping, DSMB review, unblinded review, or continued monitoring. \\
Crossing report & Predefine the descriptive effect scale reported at crossing and note that selected estimates may be inflated. \\
Non-crossing trial & State that the trial proceeds to the planned primary analysis if the e-process does not cross. \\
Multiplicity & Specify how multiple endpoints, arms, or looks across variants will be combined or adjusted. \\
\bottomrule
\end{tabular}
\end{table}

\section{Discussion}

The e-RT family comprises nonparametric sequential tests for randomized trials based on the betting framework for e-values (i-bet \cite{pmlr-v177-duan22a}). All active variants require only that treatment assignment is randomized---no distributional assumptions about outcomes are needed for validity. This makes them robust complements to model-based analyses. This manuscript focuses on binary outcomes (e-RTb), event-only monitoring (e-RTe), and continuous endpoints (e-RTc).

\subsection{Operating characteristics}

Across the simulation-validated variants, simulations demonstrate proper Type I error control, confirming the theoretical guarantee from the martingale property. Power varies by variant and scenario: for binary outcomes, approximately 50\% power for early stopping in trials designed with 80\% frequentist power, and 63--66\% in trials designed with 90\% power. The event-only variant (e-RTe) trades information for operational simplicity because non-events are not used by the e-process. The e-RTc simulations show a wager-policy tradeoff on the Cohen's $d$ scale: a matched parametric design wager can be much more powerful than the agnostic adaptive wager, while misspecification changes both power and Type M error at crossing.

When early stopping occurs across all variants, it typically happens at approximately 45--56\% of the planned sample size or event count.

These results should be interpreted as evidence for a monitoring role rather than as a replacement for the planned primary analysis. If the e-process crosses its threshold, the protocol may trigger stopping or DSMB review. If it does not cross, the trial proceeds to its planned conclusion and primary analysis.

\subsection{Traditional statistics at crossing}

When an e-RT crosses its threshold ($\geq 1/\alpha$), conventional p-values and confidence intervals computed at that same moment are selected summaries: the trial has stopped at a favorable time. They may be useful clinical descriptors, and in simulations they usually moved in the same direction as the e-process, but they should not carry the inferential claim. The e-value carries the anytime-valid evidence; traditional statistics at crossing are diagnostic and descriptive.

\subsection{What is the null hypothesis being tested?}

This approach tests whether treatment assignment can be predicted from outcomes---equivalently, whether outcomes are exchangeable between arms. Under the null, $Y_i \perp T_i$ at each observation: knowing the outcome provides no information about which arm the patient belongs to. This is neither Fisher's sharp null (every individual has exactly zero treatment effect) nor the weak null of equal population means.

Rejecting the null means outcomes predict assignments better than randomization alone would allow. This framing clarifies both the method's strength and its limitation. The strength is generality: constant effects, heterogeneous effects, and some randomization failures can all make outcomes informative. The limitation is that the current adaptive wager detects departures only when they are stable enough for a cumulative backward-looking strategy to exploit. Non-stationary effects and abrupt direction reversals are therefore blind spots for the current cumulative wager.

\subsection{Relationship to existing work}

The betting framework for hypothesis testing was developed by \citet{shafer2021}. E-values and e-processes have been extensively studied \citep{vovk2021, ramdas2022, ramdas2025}. \citet{pmlr-v177-duan22a} introduced interactive rank testing by betting (i-bet), which applies the betting framework directly to randomized experiments: an analyst sequentially bets on treatment assignments based on observed outcomes, with wealth forming a test martingale under the null. The binary approach implements this framework with a specific adaptive betting strategy tied to outcome values. Betting approaches have been established for estimating means of bounded random variables \citep{waudbysmith2023estimating}. The continuous extension adapts these principles to the two-sample randomization setting using a standardization strategy.

\citet{sokolova2026evalues} and the accompanying \texttt{evalinger} implementation \citep{evalinger2026} are especially important comparators for the present work. The first draft of e-RT was dated December 4, 2025, so the randomization-betting construction developed independently of that manuscript. This chronology is noted only to clarify the origin of e-RT, not to diminish the importance of their contribution. Their work provides a mature design-calibrated perspective on clinical-trial e-processes, particularly through growth-rate-optimal (GROW) wagers chosen from prespecified alternatives. The present work uses their contribution as a central point of comparison: e-RT is effect-size agnostic by default, whereas GROW-style design wagers encode a planned alternative to improve expected evidence growth. The two views are compatible because randomization supplies the validity engine and the wager policy supplies the efficiency profile. In this manuscript, design-calibrated wagers are therefore presented as optional efficiency tools within e-RT, not as replacements for the adaptive agnostic monitor.

\citet{koning2025} develops e-values for group invariance, including permutation tests, using batch-based likelihood ratio statistics normalized by permutation expectations.

The e-RTb shares the same martingale foundation as i-bet but differs in key respects: it operates prospectively as patients enroll rather than retrospectively on completed data; it requires no covariates or working models; and it uses betting fractions that adapt continuously to running outcome estimates rather than fixed-magnitude wagers guided by covariate-based predictions. This yields a simpler method that may be suitable for real-time trial monitoring.

\subsection{Limitations}

Several limitations should be noted. This is an experimental method under development. Simulations are not exhaustive, and the operating characteristics reported are specific to the scenarios tested.

The methods test only whether there are differences between arms; they do not provide point estimates or confidence intervals. The adaptive learning of $\hat{\delta}$ requires a burn-in period during which little evidence accumulates. For trials where parametric assumptions are plausible, model-based sequential methods will generally have better power. Our simulations used specific betting strategies; other choices may yield different operating characteristics. The betting strategy design section provides guidance on strategy selection, but optimal calibration for specific clinical scenarios remains an open question. Finally, it is unclear how these methods will behave when heterogeneity in treatment effects exists or there are temporal instabilities in effect size.

\subsection{When not to use e-RT}

The methods are not intended to replace every sequential design. They are a poor primary choice when a strong parametric model is credible and maximum model-based efficiency is the main goal; when the primary output must be an unbiased point estimate or confidence interval rather than evidence against exchangeability; when the effect is expected to reverse direction during enrollment; when informative outcome observation dominates the endpoint; or when several endpoints will be monitored without a multiplicity or e-value-combination plan. In these settings, e-RT may still be useful as a sensitivity monitor, but it should not be the only design justification without dedicated simulation.

\subsection{Conclusion}

E-processes provide anytime-valid sequential inference for randomized trials using only the guarantee of randomization. The e-RT family separates this validity engine from the wager policy used for efficiency. Its default contribution is effect-size-agnostic monitoring: a trial can be monitored continuously without specifying the treatment effect the e-process is trying to exploit. Design-calibrated wagers, including GROW-style wagers, can improve power when the design alternative is credible, but they are optional efficiency tools rather than validity requirements. This makes e-RT a conservative, transparent complement to traditional trial monitoring rather than a replacement for the final model-based analysis.

\section{Disclaimers and Version Control}

\subsection{Disclaimer}

This is an experimental method under development. Application to real patients should only be considered under surveillance from an experienced statistician and remain strongly discouraged at this point by the author. The author is not responsible for consequences of use of this method.

\subsection{LLM use statement}

Large language models were extensively used in this work. Before using LLMs, the author formulated the randomization-betting construction underlying e-RT. The author then uploaded the references in this manuscript to Gemini 3.0 Pro for brainstorming and drafting support. Subsequent versions were refined, tested, and debugged using Claude 4.5 Opus and ChatGPT 5.1 Pro. OpenAI Codex aided the Version 8 and Version 9 repository organization, simulation refactoring, wager-policy comparisons, e-RTc design-wager implementation, and manuscript cleanup.

\subsection{Acknowledgments}

The author is thankful to Aaditya Ramdas for their thoughtful comments on the first version and for pointing out previous literature to the author.

\subsection{Code and Data Availability}

The manuscript source, R implementation files, simulation scripts, generated CSV result tables, manuscript tables, and figures are maintained in the project repository: \url{https://github.com/fzampier/erandtest}. This repository is the source of truth for the computational appendix, rather than duplicating code inside the manuscript.

\subsection{Version Control}
\begin{enumerate}
    \item First Version (Dec 04, 2025)
    \item Second Version (Dec 07, 2025): Minor text adjustments; removed claim on sharp null.
    \item Third Version (Dec 11, 2025): Added continuous endpoint material; text adjustments.
    \item Fourth Version (Dec 17, 2025): Correction on signal direction on e-RTC. Text adjustments.
    \item Fifth Version (Dec 31, 2025): Added multi-state extension (e-RTms).
    \item Sixth Version (Feb 15, 2026): Added deaths-only monitoring (e-RTd). Added betting strategy design section explaining wager asymmetry across variants. Added traditional statistics at crossing discussion. Updated abstract and introduction to cover all five variants. Reordered sections.
    \item Seventh Version (Mar 08, 2026): Renamed binary e-RT to e-RTb after its introduction as the prototype. Generalized e-RTd (deaths-only) to e-RTe (event-only), broadening applicability beyond mortality. Added e-RTu (universal) section describing a domain-agnostic betting engine abstraction (under development). Updated all cross-references and discussion to reflect six variants.
    \item Eighth Version (Apr 30, 2026): Separated randomization validity from wager policy; added adaptive, design-fixed, misspecified-design, and oracle wager simulations for e-RTb/e-RTe; added Type M error at crossing diagnostics; added parametric normal-shift design wagers for e-RTc; committed simulation result tables for reproducibility; deferred earlier prototype variants from the active manuscript scope; and removed the embedded code appendix in favor of the project repository.
    \item Ninth Version (May 06, 2026): Refocused the active manuscript on e-RTb, e-RTe, and e-RTc; removed the fourth endpoint family from the active manuscript, reproducibility chain, generated artifacts, and documentation; removed parked exploratory material from the public repository; and carried forward the conditional-exchangeability clarification for event-only monitoring.
\end{enumerate}

% ============================================
% REFERENCES
% ============================================
\bibliographystyle{apalike}
\bibliography{references}

\end{document}